\documentclass[aps,twocolumn,superscriptaddress,notitlepage,pra,showpacs,nofootinbib, longbibliography]{revtex4-1}

\usepackage{amsmath,amssymb, amsthm}
\usepackage{bm}
\usepackage{graphicx}
\usepackage{braket,dsfont}
\usepackage{color}
\usepackage[10pt]{moresize}
\usepackage{mathrsfs}
\usepackage{multirow}
\usepackage{hyperref}

\usepackage{mathdots}
\let\oldcaption\caption
\renewcommand{\caption}{\sffamily \oldcaption}

\newcommand{\A}{\mathcal{A}}
\newcommand{\W}{\mathcal{W}}
\newcommand{\Mz}{\mathcal{M}_Z}
\newcommand{\Mx}{\mathcal{M}_X}
\newcommand{\N}{\mathcal{N}}
\newcommand{\mH}{H}

\newtheorem{theorem}{Theorem}
\newtheorem{lemma}{Lemma}

\newcommand{\tr}{{\rm tr}}

\renewcommand{\S}{\mathcal{S}}
\newcommand{\Q}{\mathcal{Q}}

\newcommand{\SqkM}[1]{\S_{#1}^{\mbox{\tiny $\Q$},*}}

\newcommand{\I}{\mathcal{I}}

\newcommand{\proj}[1]{\ket{#1}\!\!\bra{#1}}
\newcommand{\ketbra}[2]{\ket{#1}\!\!\bra{#2}}
\newcommand{\ExpVal}[1]{\langle{#1}\rangle}
\newcommand{\id}{\mathbb{I}}

\usepackage{soul,color}

\begin{document}
\title{Entanglement structure: entanglement partitioning in multipartite systems and its experimental detection using optimizable witnesses }


\author{He Lu}
\affiliation{Shanghai Branch, National Laboratory for Physical Sciences at Microscale and Department of Modern Physics, University of Science and Technology of China, Shanghai 201315, China}
\affiliation{CAS Center for Excellence and Synergetic Innovation Center in Quantum Information and Quantum Physics,
University of Science and Technology of China, Shanghai 201315, China}
\affiliation{CAS-Alibaba Quantum Computing Laboratory, Shanghai 201315, China}

\author{Qi Zhao}
\affiliation{Center for Quantum Information, Institute for Interdisciplinary Information Sciences, Tsinghua University, Beijing 100084, China}

\author{Zheng-Da Li}
\affiliation{Shanghai Branch, National Laboratory for Physical Sciences at Microscale and Department of Modern Physics, University of Science and Technology of China, Shanghai 201315, China}
\affiliation{CAS Center for Excellence and Synergetic Innovation Center in Quantum Information and Quantum Physics,
University of Science and Technology of China, Shanghai 201315, China}
\affiliation{CAS-Alibaba Quantum Computing Laboratory, Shanghai 201315, China}

\author{Xu-Fei Yin}
\affiliation{Shanghai Branch, National Laboratory for Physical Sciences at Microscale and Department of Modern Physics, University of Science and Technology of China, Shanghai 201315, China}
\affiliation{CAS Center for Excellence and Synergetic Innovation Center in Quantum Information and Quantum Physics,
University of Science and Technology of China, Shanghai 201315, China}
\affiliation{CAS-Alibaba Quantum Computing Laboratory, Shanghai 201315, China}

\author{Xiao Yuan}
\affiliation{Shanghai Branch, National Laboratory for Physical Sciences at Microscale and Department of Modern Physics, University of Science and Technology of China, Shanghai 201315, China}
\affiliation{CAS Center for Excellence and Synergetic Innovation Center in Quantum Information and Quantum Physics,
University of Science and Technology of China, Shanghai 201315, China}

\author{Jui-Chen Hung}
\affiliation{Department of Physics, National Cheng Kung University, Tainan 701, Taiwan}

\author{Luo-Kan Chen}
\affiliation{Shanghai Branch, National Laboratory for Physical Sciences at Microscale and Department of Modern Physics, University of Science and Technology of China, Shanghai 201315, China}
\affiliation{CAS Center for Excellence and Synergetic Innovation Center in Quantum Information and Quantum Physics,
University of Science and Technology of China, Shanghai 201315, China}
\affiliation{CAS-Alibaba Quantum Computing Laboratory, Shanghai 201315, China}

\author{Li Li}
\affiliation{Shanghai Branch, National Laboratory for Physical Sciences at Microscale and Department of Modern Physics, University of Science and Technology of China, Shanghai 201315, China}
\affiliation{CAS Center for Excellence and Synergetic Innovation Center in Quantum Information and Quantum Physics,
University of Science and Technology of China, Shanghai 201315, China}
\affiliation{CAS-Alibaba Quantum Computing Laboratory, Shanghai 201315, China}

\author{Nai-Le Liu}
\affiliation{Shanghai Branch, National Laboratory for Physical Sciences at Microscale and Department of Modern Physics, University of Science and Technology of China, Shanghai 201315, China}
\affiliation{CAS Center for Excellence and Synergetic Innovation Center in Quantum Information and Quantum Physics,
University of Science and Technology of China, Shanghai 201315, China}
\affiliation{CAS-Alibaba Quantum Computing Laboratory, Shanghai 201315, China}

\author{Cheng-Zhi Peng}
\affiliation{Shanghai Branch, National Laboratory for Physical Sciences at Microscale and Department of Modern Physics, University of Science and Technology of China, Shanghai 201315, China}
\affiliation{CAS Center for Excellence and Synergetic Innovation Center in Quantum Information and Quantum Physics,
University of Science and Technology of China, Shanghai 201315, China}
\affiliation{CAS-Alibaba Quantum Computing Laboratory, Shanghai 201315, China}

\author{Yeong-Cherng Liang}
\affiliation{Department of Physics, National Cheng Kung University, Tainan 701, Taiwan}

\author{Xiongfeng Ma}
\affiliation{Center for Quantum Information, Institute for Interdisciplinary Information Sciences, Tsinghua University, Beijing 100084, China}

\author{Yu-Ao Chen}
\affiliation{Shanghai Branch, National Laboratory for Physical Sciences at Microscale and Department of Modern Physics, University of Science and Technology of China, Shanghai 201315, China}
\affiliation{CAS Center for Excellence and Synergetic Innovation Center in Quantum Information and Quantum Physics,
University of Science and Technology of China, Shanghai 201315, China}
\affiliation{CAS-Alibaba Quantum Computing Laboratory, Shanghai 201315, China}

\author{Jian-Wei Pan}
\affiliation{Shanghai Branch, National Laboratory for Physical Sciences at Microscale and Department of Modern Physics, University of Science and Technology of China, Shanghai 201315, China}
\affiliation{CAS Center for Excellence and Synergetic Innovation Center in Quantum Information and Quantum Physics,
University of Science and Technology of China, Shanghai 201315, China}
\affiliation{CAS-Alibaba Quantum Computing Laboratory, Shanghai 201315, China}



\date{\today}





\begin{abstract}
Creating large-scale entanglement lies at the heart of many quantum information processing protocols and the investigation of fundamental physics. For multipartite quantum systems, it is crucial to identify not only the presence of entanglement, but also its detailed structure. This is because in a generic experimental situation with sufficiently many subsystems involved, the production of so-called genuine multipartite entanglement remains a formidable challenge. Consequently, focusing exclusively on the identification of this strongest type of entanglement may result in an {\em all or nothing} situation where some inherently quantum aspects of the resource are overlooked. On the contrary, even if the system is not genuinely multipartite entangled, there may still be many-body entanglement present in the system. An  identification of the entanglement structure may thus provide us with a hint on where imperfections in the setup may occur, as well as where we can identify groups of subsystems that can still exhibit strong quantum-information-processing capabilities. However, there is no known efficient methods to identify the underlying entanglement structure. Here, we propose two complementary families of witnesses for the identification of such structures. They are based, respectively, on the detection of entanglement intactness and entanglement depth, each applicable to an arbitrary number of subsystems and whose evaluation requires only the implementation of solely two local measurements. Our method is also robust against noises and other imperfections, as reflected by our experimental implementation of these tools to verify the entanglement structure of five different eight-photon entangled states. In particular, we demonstrate how their entanglement structure can be precisely and systematically inferred from the experimental measurement of these witnesses. In achieving this goal, we also illustrate how the same set of data can be classically postprocessed to learn the most about the measured system.
\end{abstract}
\maketitle

\section{Introduction}\label{Sec:Introduction}

Entanglement~\cite{Horodecki09}, one of the defining features offered by quantum theory, is known to be an essential resource in many quantum information processing tasks, including quantum computing~\cite{Nielsen:2011:QCQ:1972505}, quantum cryptography~\cite{bb84, PhysRevLett.67.661}, quantum teleportation~\cite{Teleportation93}, and the reduction of communication complexity~\cite{Buhrman:RMP} via Bell nonlocality~\cite{bell}. In the last decade, tremendous progress has been achieved in the experimental manipulation of small-scale multipartite entanglement using various physical systems~\cite{PhysRevLett.106.130506, PhysRevLett.117.210502, Lucke14prl, chen2017observation, song201710}. Indeed, a long-term goal of quantum technology is to generate medium- and eventually large-scale quantum entanglement that realizes various quantum information processing tasks.

Along this spirit, several experiments have investigated entanglement in large-scale quantum systems involving hundreds (or more) atoms~\cite{Lucke14prl, takei2016direct, Schmied2016,luo2017deterministic} or trapped ions~\cite{britton2012engineered}. However, experimentally producing large-scale genuine multipartite entanglement remains a formidable challenge owing to inevitable couplings to the environment. Consequently, an experimentally prepared $n$-partite state (for large enough $n$) typically contains only fewer-body entanglements that are segregated. To benchmark our technological progress towards the generation of large-scale genuine multipartite entanglement, it is thus essential to determine the corresponding entanglement depth~\cite{Sorensen2001}, i.e., the extent to which the prepared state is many-body entangled. Likewise, to overcome imperfections in the preparation procedure, it would be crucial to identify the extent to which the entanglements produced are segregated, as captured by the nonseparability~\cite{Horodecki09} of the state.

The identification of such entanglement structures is generally challenging, especially when full state reconstruction is infeasible. Still, the experimental preparation of a quantum resource generally follows some well-defined procedure with a well-defined target quantum state in mind. Moreover, even in the presence of experimental imperfections, such {\em a priori} knowledge of what to expect from the setup generally remains relevant. In this case, generalized entanglement witnesses (EW)  \cite{terhal2001family, guhne2009entanglement} serve as powerful alternatives for retrieving information about the underlying entanglement structure. In general,  the experimental evaluation of an EW may require the measurement of local observables that depends on the number of subsystems involved. Nevertheless, entanglement can be witnessed by a constant number of local observables~\cite{Toth2005,Knips2016}, with {\em two} being the minimum since it is impossible to distinguish entangled states from fully separable states with only one local observable. Also worth noting is the fact that the majority of the theoretical tools developed for multipartite entanglement detection~\cite{Guehne2009}  have focused exclusively on the identification of genuine multipartite entanglement, thus rendering them irrelevant in identifying the subtle entanglement structure mentioned above.

In this work, we propose two families of EWs that can respectively certify the maximum number of segregations and the minimal extent of many-body entanglement present in an $n$-qubit system, each by the measurement of \emph{solely} two local observables, i.e., the minimal possible in order to make any nontrivial conclusion. Importantly, each family of witnesses involve the same local measurement regardless of the number of subsystems present. They also do not depend on the extent of nonseparability or entanglement depth to be certified --- this follows directly from the extent to which the respective witnesses are violated. As an illustration of how these witness fare in practice, we experimentally prepare several 8-photon quantum states and demonstrate how the measurement of these two families of EWs---which involves altogether the measurement of four distinct local observables---enable us to infer nontrivial information about the underlying entanglement structure.

\section{Entanglement structure}\label{Structure}

Let $\ket{\phi}=\bigotimes_{i=1}^{m} \ket{\psi_{\mathcal{G}_i}}$ be a quantum state of $n$ parties (subsystems) divided into $m$ disjoint subsets $\{\mathcal{G}_i\}_{i=1,\ldots,m}$, each of which is described by the genuinely multipartite entangled state $\ket{\psi_{\mathcal{G}_i}}$. We say that $\{\mathcal{G}_i\}_{i=1,\ldots,m}$ fully specifies the entanglement structure of $\ket{\phi}$ as it identifies exactly all the entangled subsystems in the composite system. A partial specification of the entanglement structure can be achieved via its {\em separability}. Specifically, $\ket{\phi}$ is said to be $m$-separable~\cite{Horodecki09} ($2\le m\le n$) as it can be written as the tensor product of a pure state $\ket{\psi_{\mathcal{G}_i}}$ from $m$ disjoint subsets. The $m$-separability of a quantum state captures the notion of segregation, i.e., no physical interaction between any two subsystems from disjoint subsets is needed for the generation of $\ket{\phi}$. The larger the value of $m$, the more segregated  $\ket{\phi}$ is. Conversely, the certification that a state is non-$m$-separable implies  that $\ket{\phi}$ cannot be generated by segregating the subsystems into $m$ disjoint subsets and allowing arbitrary manipulations within each subset.

While the (non)-$m$-separability of $\ket{\phi}$ already provides us with important information about the entanglement structure of $\ket{\phi}$, it is not specifically meant to indicate the extent of many-body entanglement present in the system. {To see this, note, for example, that the four-qubit states $\ket{\chi}_{\tiny ABC}\otimes\ket{\zeta}_{D}$ and $\ket{\eta}_{\tiny AB}\otimes\ket{\tau}_{CD}$ are both 2-separable, but the generation of the former may require three-body entanglement while the latter only require up to two-body entanglement. To this end, let} us denote by  $n_i$ the number of subsystems involved in the subset $\mathcal{G}_i$ (note that $\sum_{i=1}^m n_i= n$). Then $\ket{\phi}$ is said to be $k$-producible~\cite{Guhne2005} if the largest constituent of $\ket{\phi}$ involves at most $k$ parties, i.e., if $\max_i n_i \le k$. In other words, a $k$-producible state requires at most $k$-body entanglement in its generation. Thus, the certification that a state is not $k$-producible implies that a higher level of many-body entanglement is required in its generation.

The $m$-separability and $k$-producibility of a general mixed state $\rho$ can be defined analogously: $\rho$ is $m$-separable (or $k$-producible) if it admits a convex decomposition in terms of $m$-separable ($k$-producible) pure states. Following Ref.~\cite{Sorensen2001}, we say that $\rho$ has an entanglement depth of $k$ if it is $k$-producible but not $(k-1)$-producible.  On the other hand, we say that a quantum state $\rho$ has an entanglement intactness of $m$ if it is $m$-separable but not $(m+1)$-separable. A genuinely $n$-partite entangled has an entanglement intactness (depth) of 1 ($n$), whereas a fully separable $n$-partite state has an entanglement intactness (depth) of $n$ (1). In particular, any quantum state that has an entanglement depth greater than $2$ is conventionally said to contain multipartite (many-body) entanglement.

\begin{figure*}[t!]
\centering
\includegraphics[width=2\columnwidth]{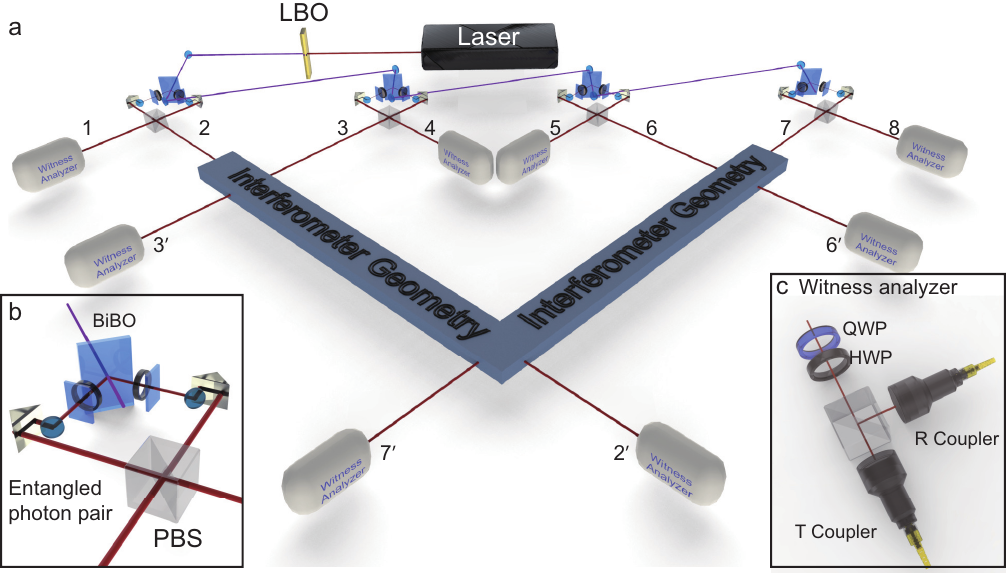}
\caption{\textbf{Schematic showing the experimental  setup. (a) }  A global view of the experimental setup used to generate different eight-photon entangled states. A pulse from a pulsed Ti-sapphire laser (with a central wavelength of 780~nm, a duration of 130~fs, and an average power of 3.5~W) passes through a frequency doubler, by which it is changed to  a UV pulse with  a central wavelength of 390 nm and  an average power of 1.3~W. Then, the UV pulse is directed by reflective mirrors to shine on four BiBO crystals successively. Shining  a UV pulse on  a BiBO crystal will generate  (probabilistically) a photon pair maximally entangled in the polarization degree of freedom via spontaneous parametric down-conversion (SPDC).  Photons in path modes 2, 3, 6, and 7 are  then injected into an interferometric network to generate $\ket{G_8}$, $\ket{G_{62}}$, $\ket{G_{44}}$, $\ket{G_{422}}$ and $\ket{G_{2222}}$ by the corresponding interferometric geometry settings. Finally, eight photons are analyzed by witness analyzers  via single-photon detectors. To suppress  the higher-order emission in SPDC, we attenuate the average power of the UV pulse to 500~mW. The eight-fold coincidences we observed in creating $\ket{G_8}$, $\ket{G_{62}}$, $\ket{G_{44}}$, $\ket{G_{422}}$ and $\ket{G_{2222}}$ are 8/h, 20/h, 20/h, 36/h and 70/h, respectively due to different postselection probabilities (see Fig.~\ref{fig:geometry} for further explanation and Appendix~\ref{sourceappendix} for the actual number of eight-fold coincidences registered in each case.  \textbf{(b) }the experimental setup  used to generate maximally entangled photon pairs. More details concerning  the generation of entangled photons can be found in Appendix~\ref{sourceappendix}. \textbf{(c) }the witness analyzer. An arbitrary observable $\mathcal O$ can be represented as $\mathcal O=\ket{i_{+1}}\bra{i_{+1}}-\ket{i_{-1}}\bra{i_{-1}}$, where $\ket{i_{\pm1}}$ is the eigenstate of $\mathcal O$ with eigenvalue of $\pm1$. The combination of quarter-  (QWP) and half- (HWP) waveplates as well as a polarization beam splitter (PBS) makes $\ket{i_{+1}}$ click on the transmissive detector and $\ket{i_{-1}}$ click on the reflective detector. }
\label{fig:setup}
\end{figure*}

We are now in the position to introduce our witnesses for entanglement intactness. To certify the nonseparability and hence an upper bound on the entanglement intactness of a given quantum state, we introduce the following 2-parameter family of two-observable witnesses:
\begin{equation}\label{eq:separability}
\W^n_{se}(\alpha)=\alpha \Mz+\Mx\!\!\stackrel{m-\text{sep.}}{\le}\!  {\id_n}\, \max\{\alpha, \tfrac{\alpha}{2^{m-1}}+1\},
\end{equation}
where $\alpha\in(0, 2]$ is a free parameter, $\Mz=\left(\proj{0}\right)^{\otimes{n}}+\left(\proj{1}\right)^{\otimes{n}}$ and $\Mx=\sigma_x^{\otimes{n}}$ are $n$-qubit observables, $\sigma_x=\ketbra{0}{1}+\ketbra{1}{0}$ is the Pauli $x$ matrix, $\{\ket{0},\ket{1}\}$ are the computational basis states, $\id_n$ is the $2^n\times 2^n$ identity matrix, and $m$-sep.~in Eq.~\eqref{eq:separability} signifies that the inequality holds true at the level of the expectation value for all $m$-separable states. In other words, for an arbitrary $n$-partite state $\rho$, if $\langle \W^n_{se}(\alpha) \rangle_\rho>\max\{\alpha, \frac{\alpha}{2^{m-1}}+1\}$, one certifies that $\rho$ has an entanglement intactness of $m-1$ or lower. To ease notation, we  abbreviate $\W^n_{se}(\alpha=2)$ as $\W_n$.

For witnessing entanglement depth, inspired by Ref.~\cite{PhysRevLett.114.190401}, we introduce the following family of witnesses, which also involve only two local measurements
[but (possibly) in a basis different from those of $\W^n_{se}(\alpha)$]:
\begin{equation}\label{Eq:depth}
	\W^n_{de}(\gamma)=\gamma\kappa^n\A-\A' \stackrel{k-\text{prod.}}{\le}   {\id_n}\,\beta_{n, k}(\gamma)
\end{equation}
where $\gamma\in(0,2]$ is a free parameter,  $\A = (\frac{\A_-+\A_+}{2\kappa})^{\otimes n}$, and $\A'=(\A_+)^{\otimes n}$ are $n$-partite $\pm1$-valued observables, $\A_{\pm}$ is a single-partite $\pm1$-valued observable, $\kappa\neq0$ (which holds for $\A_+\neq \A_-$) is a normalization constant, and $\beta_{n, k}(\gamma)$ is the $k$-producible bound of the $n$-partite version of the witness $\W^n_{de}(\gamma)$. In Eq.~\eqref{Eq:depth}, $k$-prod. signifies that the inequality holds true at the level of the expectation value for all $n$-partite $k$-producible states. In other words, for an arbitrary $n$-partite state $\rho$, if $\langle \W^n_{de}(\gamma) \rangle_\rho> \beta_{n,k}(\gamma)$, one certifies that $\rho$ has an entanglement depth of at least $k+1$. For $\gamma=2$ and if no further assumption (including the underlying Hilbert space dimension) is made on $\A_\pm$, it follows from Ref.~\cite{PhysRevLett.114.190401} that $\beta_{n, 1}=1$, $\beta_{n,2}=\sqrt{2}$, $\ldots$ for all $n\ge 2$. For specific choices of qubit observables $\A_{\pm}$, these bounds can be tightened to provide better noise robustness (see Sec.~\ref{Sec:Exp}, specifically our elaboration in page~\pageref{Pg:Specific}). Despite their simplistic form, the derivation of the bounds for the two families of witnesses is highly nontrivial and may serve as a basis for the derivation of other entanglement witnesses. For details, see Appendix~\ref{theoryappendix}.

A few other remarks are now in order. First, in contrast with ordinary entanglement witnesses, we see from Eqs.~\eqref{eq:separability} and~\eqref{Eq:depth} that the measured value for these witnesses is precisely the information that we need to provide further details about the underlying entanglement structure. Moreover, both families of witnesses involve a free positive parameter that may be optimized {\em a posteriori} to identify the best possible  upper (lower) bound on the entanglement intactness (depth) of $\rho$. Finally, it is worth noting that these witnesses can be easily adapted to an arbitrary choice of local basis, i.e., even after applying an arbitrary local unitary transformation on each qubit, the $m$-separable bounds and the $k$-producible bounds of the transformed witnesses evidently remain unchanged. For example, via the local unitary transformation $\id_2^{\otimes n-1}\otimes\sigma_z$, the witness $\W^{n}_{se}(\alpha)$ gets transformed to $\W^{n'}_{se}(\alpha)=\alpha \Mz-\Mx$, which complements $\W^{n}_{se}(\alpha)$ in detecting the many-body entanglement present in a larger set of quantum states, see Appendix~\ref{appendixseparablity} for details.\footnote{To decide if such a pretransformation is relevant, one may want to allow the measurement of an additional local observable, such as $\sigma_y$, by one of the parties (or an appropriate modification thereof) before measuring the witness operator itself.}

\section{Experimental realization}\label{Sec:Exp}

\begin{figure*}[t!]
\centering
\includegraphics[width=2\columnwidth]{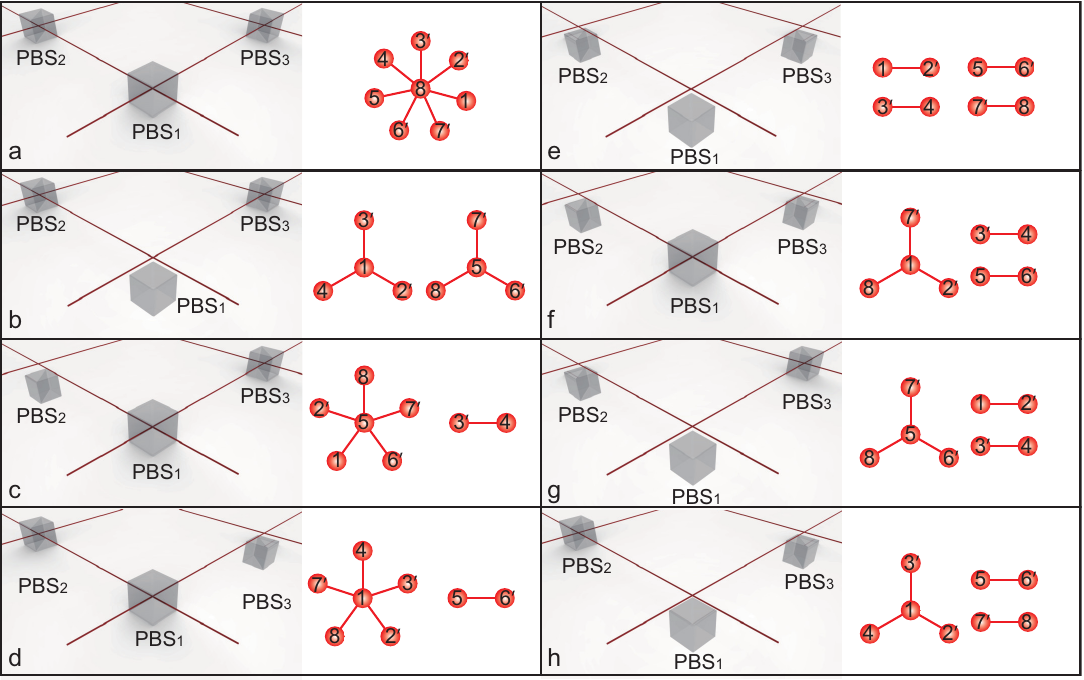}
\caption{\textbf{Interferometric geometries leadings to different entanglements.}
Here, $\ket{\textrm{GHZ}_n}$ is an example of a graph state and can  be represented by a star graph. Such a representation makes its preparation procedure evident: Each node represents a photon prepared in the state $\ket{+}=\tfrac{1}{\sqrt{2}}(\ket{H}+\ket{V})$ and each edge joining two nodes represents a controlled-$Z$ operation performed between the corresponding photons. The number inscribed in each node labels the path mode of the photons prepared in our experiment. The PBS is fixed on a lifting platform. By adjusting the lift height, we can switch the state of the PBS between \emph{up} and
\emph{down}. {For each PBS set to \emph{up}, since we postselect the cases where photons exit from both output ports of the PBS, the count rate reduces approximately by half compared with the case when the PBS is set to \emph{down}.}}
\label{fig:geometry}
\end{figure*}

Experimentally, we use the polarization degree of freedom to encode the qubit state $\ket{H(V)}=\ket{0(1)}$, where $H(V)$ denotes the horizontal (vertical) polarization. The experimental setup used to generate  various entangled states is shown in Fig.~\ref{fig:setup}(a). We first generate four pairs of maximally entangled states $1/\sqrt{2}(\ket{H_iH_j}+\ket{V_iV_j})$ by shining an ultraviolet (UV) pulse successively on four BiB$_{3}$O$_{6}$ crystals [as shown in Fig.~\ref{fig:setup}(b)] with $i$, $j$ denoting the path modes (more details are shown in Fig.~\ref{fig:setup_details} of Appendix~\ref{sourceappendix}).  Photons in path modes 2, 3, 6, and 7 are then injected into an interferometric network (IN), which consists of three  polarization beam splitters (PBSs)  with four input and output ports (as shown in Fig.~\ref{fig:geometry}).  Each PBS is controlled by an individual lifting platform that can be set to either the \emph{up} or \emph{down} state. When  a PBS is in the \emph{up} state, it facilitates the interference of the two photons arriving at its two input ports. On the contrary, there is no interference between the incoming photons when the PBS is  in the \emph{down} state. With three independently controlled PBSs, one can construct eight interferometric geometries, which correspondingly lead to eight possible photonic entangled states that fall under five distinct entanglement structures.

Let $\ket{\textrm{GHZ}_n}=\frac{1}{\sqrt{2}}(\ket{H}^{\otimes n}+\ket{V}^{\otimes n})$ denote an $n$-photon Greenberger-Horne-Zeilinger (GHZ) state. Then, in the absence of imperfection, the interferometer can thus be used to produce the five different entangled states (one from each entanglement structure): $\ket{G_8}, \ket{G_{62}}, \ket{G_{44}}, \ket{G_{422}}, \ket{G_{2222}}$ where the subscripts  $i_1\ldots i_m$ of $\ket{G_{i_1\ldots i_m}}=\bigotimes_{j=1}^m \ket{\textrm{GHZ}_{i_j}}$ are used to label the entanglement structure. For example, when the  states of $\text{PBS}_1$, $\text{PBS}_2$, and $\text{PBS}_3$  are set, respectively, to \emph{up}, \emph{up} and \emph{down}, the  corresponding interferometric geometry is depicted in Fig.~\ref{fig:geometry}(d).  The interaction of photons 2 and 3 with $\text{PBS}_2$ leads  to the state $\ket{\textrm{GHZ}_4}=\frac{1}{\sqrt{2}}(\ket{H}^{\otimes 4}+\ket{V}^{\otimes 4})_{12^{\prime}3^{\prime}4}$. On the other hand, since $\text{PBS}_3$ plays no role in the path of photon 6 and photon 7,  the  outgoing state of photons 5-6$^{\prime}$-7$^{\prime}$-8 is a  tensor product of $\ket{\textrm{GHZ}_2}_{56^{\prime}}$ and $\ket{\textrm{GHZ}_2}_{7^{\prime}8}$. Finally, the interaction at $\text{PBS}_1$  by the incoming photon at modes 2$^{\prime}$ and 7$^{\prime}$ leads to an entanglement in the form of $\ket{G_{62}}=\frac{1}{2}(\ket{H}^{\otimes 6}+\ket{V}^{\otimes 6})_{12^{\prime}3^{\prime}47^{\prime}8}\otimes (\ket{H}^{\otimes 2}+\ket{V}^{\otimes 2})_{56^{\prime}}$. More details of the state preparation procedure are shown in Appendix~\ref{experimentappendix}.

Note that  the geometries depicted in Fig.~\ref{fig:geometry}(c) and Fig.~\ref{fig:geometry}(d) produce essentially the same entangled state as $\ket{G_{62}}$ but differ in their path mode. Similarly, the entanglement produced in Fig.~\ref{fig:geometry}(f), Fig.~\ref{fig:geometry}(g) and Fig.~\ref{fig:geometry}(h)  is essentially the same as that of $\ket{G_{422}}$. In our experiment, we choose the geometries in Figs.~\ref{fig:geometry}(a)-\ref{fig:geometry}(c), ~\ref{fig:geometry}(e) and~\ref{fig:geometry}(h) for the preparation of  five different entangled states. The generated 8-photon entanglement is detected and analyzed by eight witness analyzers in paths 1, 2$^{\prime}$, 3$^{\prime}$, 4, 5, 6$^{\prime}$, 7$^{\prime}$, and 8. As shown in Fig.~\ref{fig:setup}(c), a witness analyzer consists of a quarter-wave plate (QWP), a half-wave plate (HWP), a PBS and two single-photon detectors.

In reality, there are always imperfections in the setup, and the entangled state produced is thus more aptly described by a density matrix $\rho$. For ease of comparison, in a setup used to produce the quantum state $\ket{G_i}$, we shall denote the actual quantum state produced by $\rho_i$. To determine the entanglement structure of $\rho_i$, we thus begin by measuring the expectation value of the observables $\mathcal M_X= {\sigma_x}^{\otimes 8}$ and $\mathcal M_Z= \left(\proj{H}\right)^{\otimes 8}+\left(\proj{V}\right)^{\otimes 8}$.  Whenever the observed $\langle \mathcal M_Z\rangle$ and $\langle \mathcal M_X\rangle$ violate the inequality [corresponding to Eq.~\eqref{eq:separability} with $\alpha=2$]
\begin{equation}\label{eq:genuinewitness}
	2\langle \mathcal M_Z\rangle + \langle \mathcal M_X\rangle  \stackrel{\text{2-sep.}}{\le} 2,
\end{equation}
we can thus conclude that $\rho_i$ exhibits genuine eight-photon entanglement. As shown in Fig.~\ref{fig:data}(a), we observe that ($\langle \mathcal M_Z\rangle, \langle \mathcal M_X\rangle)=(0.80(2), 0.63(4)$) on  $\rho_8$ which violates Eq.~\eqref{eq:genuinewitness}, while the corresponding expectation values for $\rho_8$, $\rho_{62}$, $\rho_{44}$, $\rho_{422}$, $\rho_{2222}$, as summarized in Table~\ref{Tbl:SepRes},  satisfy Eq.~\eqref{eq:genuinewitness}. The results indicate that $\rho_8$ is genuinely eight-photon entangled, but the entanglement structure of the rest
cannot be concluded from the witness of Eq.~\eqref{eq:genuinewitness}.

\begin{figure*}[t!]
\centering
\includegraphics[width=2\columnwidth]{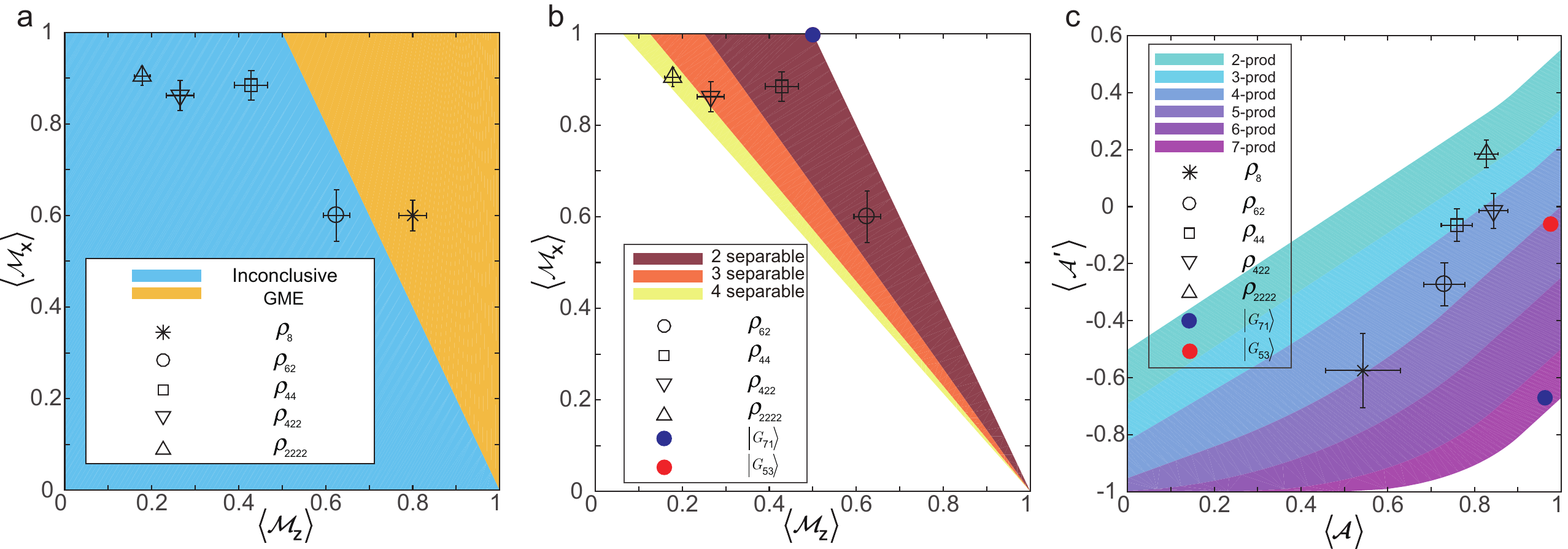}
\caption{\textbf{Experimental results for witnesses certifying genuine multipartite entanglement (GME), entanglement intactness and entanglement depth. (a)} GME revealed by the measurement of the witness of Eq.~\eqref{eq:genuinewitness} via $\ExpVal{\Mz}$ and $\ExpVal{\Mx}$. \textbf{(b)} Entanglement intactness revealed by the measurement of $\ExpVal{\W_8^{se}(\alpha)}$ via $\ExpVal{\Mz}$ and $\ExpVal{\Mx}$ for judicious choice of $\alpha$. \textbf{(c)} Entanglement depth revealed by the measurement of $\ExpVal{\W_8^{de}(\gamma)}$ via $\ExpVal{\A}$ and $\ExpVal{\A'}$ for judicious choice of $\gamma$. For comparison, we have also included here with the red (blue) dot the theoretical value of the witness for an ideal $\ket{G_{53}}$ ($\ket{G_{71}}$) assuming the same set of measurements.}
\label{fig:data}
\end{figure*}

\begin{table}[h!]
\centering
\begin{tabular}{|c|c|c|c|c|c|c|}
\hline
 State & $m$ & $\alpha$ & $\ExpVal{\Mz}$ & $\ExpVal{\Mx}$ & $\ExpVal{\W_{se}^{8}\left(\alpha\right)}$ & Ent. Intactness
\\ \hline
$\rho_8$ & 2 & 2 & 0.80(2) &  0.63(4) & 2.23 (3) & 1\\
$\rho_{62}$ & 3 & 4/3 & 0.63(3) & 0.60(5) & 1.43(7) & $\le 2$\\
$\rho_{44}$ & 3 & 4/3 & 0.43(4) & 0.89(3) & 1.46(6) & $\le 2$ \\
$\rho_{422}$ & 4 & 8/7 & 0.27(3) & 0.86(3) & 1.17(5) & $\le 3$\\
$\rho_{2222}$ & 5 & 16/15 & 0.18(2) &  0.91(2) & 1.09(3) & $\le 4$ \\ \hline
$\ket{G_{71}}$  & 3 & 4/3 & 0.5 &  1 & 5/3 & $\le 2$ \\
$\ket{G_{53}}$ & 3 & 4/3 & 0.5 &  1 & 5/3 & $\le 2$ \\
\hline
\end{tabular}
\caption{Summary of our experimental results for determining the entanglement intactness of the prepared state $\rho_i$. The second and the third column give our choice of the free parameters $m$, $\alpha$ for the witness described in Eq.~\eqref{eq:genuinewitness}. The experimentally measured expectation values $\ExpVal{\Mz}$, $\ExpVal{\Mx}$ and  $\ExpVal{\W_{se}^{8}\left(\alpha\right)}=\alpha\ExpVal{\Mz}+\ExpVal{\Mx}$ are given in the next three columns. The last column gives our best upper bound on the entanglement intactness of $\rho_i$ based on the measured value of $\W_{se}^{8}\left(\alpha\right)$.  An entanglement intactness of $m$ means that the state cannot be produced by segregating the subsystems into $m+1$ or more groups. The errors are deduced from propagated Poissonian counting statistics of the raw photon detection events (see Appendix~\ref{sourceappendix} for raw data). For comparison, we have also included here the theoretical values of the witness for an ideal $\ket{G_{53}}$ and $\ket{G_{71}}$ assuming the same set of measurements.\label{Tbl:SepRes}}
\end{table}

Note, however, that Eq.~\eqref{eq:genuinewitness} only represents a specific case  ($\alpha=m=2$) of the family of witnesses considered in Eq.~\eqref{eq:separability}.  Further nontrivial information on the entanglement structure, specifically the $m$-separability of $\rho_i$ can also be deduced from the measured value of $\Mz$ and $\Mx$. Specifically, by varying $\alpha\in(0,2]$, one can identify the smallest value of $m=2,3,\ldots,8$ whereby the witnesses of Eq.~\eqref{eq:separability} are violated;\footnote{Of course, for a very poorly prepared system, it could happen that for all $2\le m\le n$, none of the witnesses from Eq.~\eqref{eq:separability} is violated.} this minimum value of $m$ then provides an upper bound of $m-1$ on the entanglement intactness of the measure system. To this end, it is worth noting that both $\W_{se}^{8}\left(\alpha\right)$, $\W_{se}^{8'}\left(\alpha\right)$ and their $m$-separable bounds are linear in $\alpha$. For any given value of $m$, the optimal choice of the free parameter $\alpha$ in Eq.~\eqref{eq:separability} is obtained by setting $\alpha=\tfrac{\alpha}{2^{m-1}}+1$, thereby giving $\alpha=\tfrac{2^{m-1}}{2^{m-1}-1}$, e.g., $\alpha=2, \tfrac{4}{3}, \tfrac{8}{7},$ and $\tfrac{16}{15}$, respectively, for $m=2,3,4,$ and 5. Note that in each of these cases, the value of $\alpha$ is precisely the $m$-separable bound given in Eq.~\eqref{eq:separability}. A direct comparison between the measured value of $\W_{se}^{8}\left(\alpha\right)$ and the various $m$-separable bounds then allows us to determine an upper bound on the entanglement intactness of the measured system.

A graphical illustration of these fine-tuned $m$-separability witnesses pcorresponding to the blue region in Fig.~\ref{fig:data}(a)] is shown in Fig.~\ref{fig:data}(b).
Once the observed expectation values $(\langle\mathcal M_Z\rangle, \langle\mathcal M_X\rangle)$  of $\rho_i$ are found to lie in the $(m-1)$-separable region, it violates the  witness of Eq.~\eqref{eq:separability} for  $m$-separability, thereby certifying that $\rho_i$ has an entanglement intactness of $m-1$ or lower. Equivalently, we see from Table~\ref{Tbl:SepRes} that the entanglement intactness of $\rho_{62}$, $\rho_{44}$, $\rho_{422}$, and $\rho_{2222}$ is upper bounded, respectively, by  2, 2, 3, and 4, which matches exactly with that of $\ket{G_{62}}$, $\ket{G_{44}}$, $\ket{G_{422}}$, and $\ket{G_{2222}}$.

\label{Pg:Specific}While these bounds on entanglement intactness already shed some light on the underlying entanglement structure, they are not yet informative enough to suggest any specific entanglement structure associated with the measured system. To this end, we also measure the witnesses for entanglement depth  given in Eq.~\eqref{Eq:depth} by measuring the  expectation value of $\A$ and $\A'$ for each of the prepared states. Specifically, based on the data that we have collected in measuring $\ExpVal{\mathcal M_Z}$ and $\ExpVal{\mathcal M_X}$ and the ansatz $\A_{\pm}= \cos\theta_\pm \sigma_x +\sin\theta_\pm\sigma_y$ (see Appendix~\ref{App:Depth}), a reasonably good choice of qubit observables appears to be those corresponding to $\theta_\pm=\tfrac{3(1\pm 8)}{80}$ (thereby giving $\kappa=\cos\tfrac{3}{10}$), where $\sigma_y=-i\ketbra{0}{1}+i\ketbra{1}{0}$ is the Pauli $y$ matrix.

The corresponding tightened $k$-producible bound $\beta_{8,k}(\gamma)$ as a function of $k$ and $\gamma$ is given in the Appendix~\ref{App:Depth}. Our experimental results shown in Fig.~\ref{fig:data}(c), and summarized in Table~\ref{Tbl:DepthRes} allow us to conclude a lower bound on entanglement depth of 4, 4, 3, 4, and 2, respectively, for the state $\rho_8$, $\rho_{62}$, $\rho_{44}$, $\rho_{422}$, and $\rho_{2222}$. Evidently, only the measurements of $\ExpVal{W_{de}^{8}}$ for $\rho_{422}$ and $\rho_{2222}$ reveal the expected entanglement depth, while the lower bound on entanglement depth obtained for the other states is clearly suboptimal. Our separate analysis shows that this is caused by the undesired noises in our experiment, specifically the higher-order emissions in SPDC and the mode mismatch of the interference. We analyze the decoherence induced by these two noises in Appendix~\ref{noiseappendix}.

\begin{table}[h!]
\centering
\begin{tabular}{|c|c|c|c|c|c|c|c|}
\hline
 State & $k$ &  $\gamma$ &  $\beta_{8, k}(\gamma)$ & $\ExpVal{\A}$ & $\ExpVal{ \A'}$ & $\ExpVal{\W_{de}^{8}(\gamma)}$  & Ent. Depth
\\ \hline
{$\rho_{8}$}   & 3 & 2 &  1.1699 & 0.54(9) & -0.57(9) &  1.32(15)& $\ge 4$\\
$\rho_{62}$ & 3 & 2 &  1.1699 & 0.73(5) & -0.27(8) &  1.29(8) & $\ge 4$\\
$\rho_{44}$ & 2 & 8/5 & 0.7904 & 0.76(3) & -0.07(5) & 0.91(7)  & $\ge 3$ \\
$\rho_{422}$ & 3 & 8/5 &  0.9137 & 0.84(3) & -0.02(6) & 0.95(7)  & $\ge 4$\\
$\rho_{2222}$ & 1 & 2 & 0.8365 & 0.83(3) &  0.19(5) & 0.96(5) & $\ge  2$ \\ \hline
$\ket{G_{71}}$ & 6 & 2 & 1.8858 & 0.9651 &  -0.6714 & 2.0106 & $\ge  7$ \\
$\ket{G_{53}}$ & 4 & 2 & 1.3856 & 0.9763 &  -0.0617 & 1.4164 & $\ge  5$ \\
\hline
\end{tabular}
\caption{Summary of our experimental results for determining the entanglement depth of the prepared state $\rho_i$. The second and the third column give our choice of the free parameters $k$ and $\gamma$ in Eq.~\eqref{Eq:depth}. The fourth column gives the value of the corresponding $k$-producible bound. The experimentally measured expectation values $\ExpVal{\A}$, $\ExpVal{\A'}$ and  $\ExpVal{\W_{de}^{8}(\gamma)}=\gamma\kappa^8\ExpVal{\A}-\ExpVal{\A'}$ are given in column 5-7. The last column gives our best lower bound on the entanglement depth of $\rho_i$ based on the measured value of $\W_{de}^{8}(\gamma)$. An entanglement depth of $k$ means that the state requires at least $k$-body entanglement for its preparation. For comparison, we have also included here the theoretical values of the witness for an ideal $\ket{G_{53}}$ and $\ket{G_{71}}$ assuming the same set of measurements.}\label{Tbl:DepthRes}
\end{table}

Nevertheless, as we demonstrate below, the measurement results obtained from both witnesses are useful and complement each other nicely---at least in our setup---to {\em suggest} some {\em minimal} entanglement structure of the measured system. By {\em minimal}, we mean that the corresponding entanglement structure is compatible with {\em all} the empirical observation, and it is also not more entangled nor more complicated than necessary to explain these empirical observations. Thus, despite the fact that a general mixed state does not have a unique convex decomposition, we are only concerned with identifying a compatible entanglement structure that is not  a convex mixture of different entanglement structures.

Coming back to the identification of a minimal entanglement structure associated with our setup, suppose that the IN is controlled by three binary random number generators, each of which determines the state of one of the PBSs. Thus, the IN is randomly set to be one of the geometries depicted in Fig.~\ref{fig:geometry}, thereby resulting in one of the corresponding entanglement structures. In the next two paragraphs, we show how to deduce the structure corresponding to $\ket{G_{422}}$ and $\ket{G_{2222}}$. The corresponding analysis for $\ket{G_{62}}$ and $\ket{G_{44}}$ is shown in Appendix~\ref{App:OtherStructures}.
To this end, we use a circular chart to schematically represent the entanglement structure of the underlying state, see Fig.~\ref{fig:data2}.  {\em A priori}, the chart is split into eight equal pieces, where each piece represents one of the subsystems (a photon) labeled uniquely by their path:   $\{1,2^{\prime}, 3^{\prime}, 4, 5, 6^{\prime}, 7^{\prime}, 8\}$. Our goal is to determine a minimal entanglement structure compatible with the empirical observation. If any of the subsystems is found to be genuinely multipartite entangled, we combine the respective pieces and color them the same way.

\begin{figure*}[t!]
\centering
\includegraphics[width=2\columnwidth]{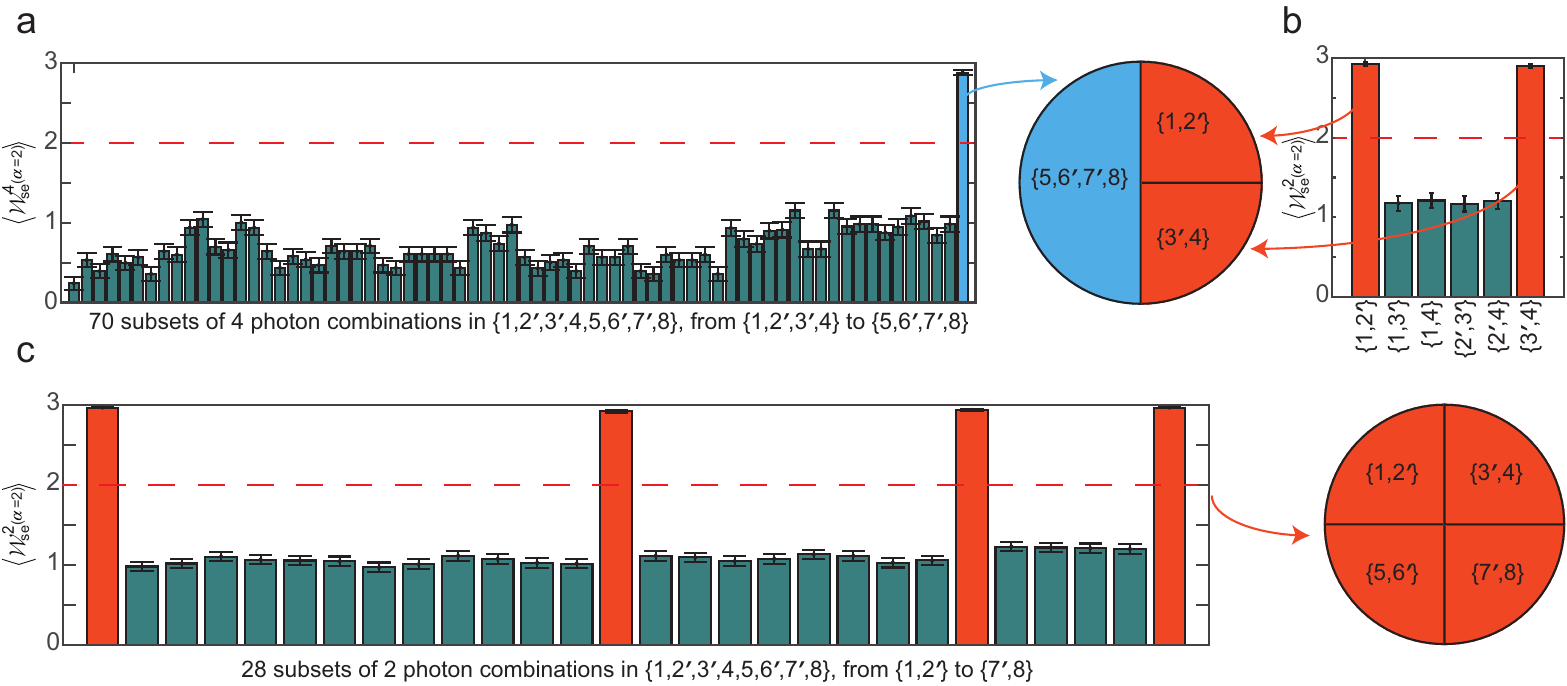}
\caption{\textbf{Experimental results leading to the determination of entanglement structure. (a) } Expectation value of the four-partite GME witness $\W_4=2\left[\left(\proj{H}\right)^{\otimes4}+\left(\proj{V}\right)^{\otimes4}\right]+\sigma_X^{\otimes 4}$ for all four-party subsystems of $\rho_{422}$. \textbf{(b)} Expectation value of the two-party GME witness $\W_2=2\left[\left(\proj{H}\right)^{\otimes2}+\left(\proj{V}\right)^{\otimes2}\right]+\sigma_X^{\otimes 2}$ for all two-party subsystems among the path modes $\{1, 2^{\prime}, 3^{\prime}, 4\}$ of $\rho_{422}$. \textbf{(c) }  Expectation value of the two-party GME witness $\W_2$ for all two-party subsystems of $\rho_{2222}$.}
\label{fig:data2}
\end{figure*}

For example, since our measurement of $\ExpVal{\W^8_{de}(\gamma)}$ on $\rho_{422}$ witnesses an entanglement depth of 4 or more, at least four of the photons exhibit GME. Importantly, our measurement results of $\ExpVal{\W^8_{se}(\alpha)}$ also allow us to evaluate $\ExpVal{\W^k_{se}(\alpha)}$ among any $k$-partite subset of the 8 photons. Indeed, from the measured expectation values of the four-partite witness $\W^4_{se}(\alpha=2)$ [see Fig.~\ref{fig:data2}(a)], only the four photons with path modes $\{5, 6^{\prime}, 7^{\prime}, 8\}$ seem to be genuinely four-photon entangled. On the other hand, our measurement of $\ExpVal{\W^8_{se}(\alpha=8/7)}$ concludes that $\rho_{422}$ is at most triseparable. Combining this with the above observation suggests that $\rho_{12^{\prime}3^{\prime}4}$ is a biseparable state. Thus, $\rho_{12^{\prime}3^{\prime}4}$ can be either a tensor product of a genuinely three-photon entangled state and a single-photon state, or a tensor product of two two-photon entangled states. Our evaluation of $\W^2_{se}(\alpha=2)$  for all possible combinations of two photons from path modes $1,2^{\prime},3^{\prime}$, and 4 [see Fig.~\ref{fig:data2}(a)] clearly reveals that the two photons from path $1, 2^{\prime}$, as well as those from path $3^{\prime}, 4$ are entangled. At the same time, our measurement of $\ExpVal{\W^3_{se}(\alpha=2)}$ among all possible three-photon combinations from $1,2^{\prime},3^{\prime}$, and 4 does not reveal any three-photon entanglement. The above observations, together with the assumptions stated above, lead us to conclude that $\rho_{422}$  shares the same {(minimal)} entanglement structure as $\ket{\textrm{GHZ}_4}_{56^{\prime}7^{\prime}8}\otimes\ket{\textrm{GHZ}_2}_{12^{\prime}}\otimes\ket{\textrm{GHZ}_2}_{3^{\prime}4}$.

Similarly, for $\rho_{2222}$, our measurement of $\ExpVal{\W^8_{de}(\gamma)}$ and $\ExpVal{\W^8_{se}(\alpha)}$ leads to the conclusion that $\rho_{2222}$ involves at least 2-body entanglement while {\em not} being 5-separable. Various entanglement structures are compatible with these observations. However, from the observed values of $\ExpVal{\W^2_{se}(\alpha=2)}$ for all possible two-photon combinations [see Fig.~\ref{fig:data2}(c)], we see that the photon pairs from path modes $\{1, 2^{\prime}\}$, $\{3^{\prime}, 4\}$, $\{5, 6^{\prime}\}$, and$\{7^{\prime}, 8\}$ are clearly entangled. Thus, with the assumptions stated above, the only entanglement structure compatible with these observations is that of $\ket{\textrm{GHZ}_2}_{12^{\prime}}\otimes\ket{\textrm{GHZ}_2}_{3^{\prime}4}\otimes\ket{\textrm{GHZ}_2}_{56^{\prime}}\otimes\ket{\textrm{GHZ}_2}_{7^{\prime}8}$.

At this point, one may wonder whether our experimental setup is capable of generating 8-photon entangled states with other entanglement structures (such as those involving odd-party entangled states), and how our witnesses fare in those cases. Let us remark that our experimental setup is not limited  to generating only the even-party-entangled state depicted in Fig.~\ref{fig:geometry} ---it can, in principle be used to produce {\em all} quantum states of the form $\ket{G_{i_1\ldots i_m}}$ (such as those containing only odd-party-entangled states). However, we did not experimentally prepare these other states as their generation (using our setup) involves heralding and thus a significantly lower count rate. As an example, to create an 8-photon entangled state $\ket{G_{71}}$ that is only 7-photon entangled, we can set two polarizers in paths 1 and 2 at 45$^\circ$, which project $\ket{\text{GHZ}_2}_{12}$ on state $\ket{++}_{12}$. Then, by raising all three PBSs, the state $\ket{G_{71}}=\frac{1}{2}(\ket{H}^{\otimes 7}+\ket{V}^{\otimes 7})_{2^{\prime}3^{\prime}456^{\prime}7^{\prime}8}\otimes (\ket{H}+\ket{V})_{1}$ can be obtained.  Likewise, to generate $\ket{G_{53}}$, we can set a polarizer between $\text{PBS}_1$ and $\text{PBS}_3$ at 45$^\circ$ in the geometry shown in Fig.~\ref{fig:geometry}a. Then, the state $\ket{G_{53}}=\frac{1}{2}(\ket{H}^{\otimes 5}+\ket{V}^{\otimes 5})_{12^{\prime}3^{\prime}47^{\prime}}\otimes((\ket{H}^{\otimes 3}+\ket{V}^{\otimes 3})_{56^{\prime}8}$ can be generated.

To illustrate that our theoretical methods also work well in these cases, we have calculated $\ExpVal{\mathcal M_Z}$,  $\ExpVal{\mathcal M_X}$, and $\ExpVal{\mathcal A}$, $\ExpVal{\mathcal A^{\prime}}$ for $\ket{G_{53}}$ and $\ket{G_{71}}$, respectively. The results are given in Tables~\ref{Tbl:SepRes} and~\ref{Tbl:DepthRes} and shown as blue and red dots in Fig.~\ref{fig:data}(b) and Fig.~\ref{fig:data}(c). Together, they confirm that $\ket{G_{53}}$ is at most a biseparable state with an entanglement depth of at least 5, and $\ket{G_{71}}$ is at most a biseparable state with an entanglement depth of at least 7. By systematically evaluating the value of the witnesses for the right number of parties over all possible choices of parties---as we did above---and finding out which subset of photons is genuinely multipartite entangled, we obtain a compatible entanglement structure, which coincides with that of the given state.

\section{Conclusion}

In this work, we introduce the notion of an {\em entanglement structure}, which details not only the extent of many-body entanglement present but also their segregation among the various subsystems. Identifying the entanglement structure of an arbitrary multipartite quantum state, as with the certification of genuine multipartite entanglement, generally requires an exponential number of local measurements. Nonetheless, the retrieval of any partial information on the entanglement structure of an experimentally-prepared system is always welcome, as it provides diagnostic information on where imperfections in the setup may lie. Importantly, such information is often already available in the data collected for the measurement of entanglement witnesses, even if the measured value does not reveal genuine multipartite entanglement. Here, we propose two complementary families of witnesses capable of bounding, respectively, the entanglement intactness (i.e., nonseparability) and the entanglement depth of the measured system, thereby providing nontrivial information about the underlying entanglement structure.

Our scheme works for any number of parties and can be generalized to arbitrary dimensions~\cite{QiXiao:Unpublished}. In contrast with conventional entanglement verification schemes, our witnesses involve free parameters that can be varied {\em a posterori}, thereby allowing us to optimize---in a similar spirit to Ref.~\cite{Eisert:2007}---the data collected to arrive at the strongest possible conclusion.  Note also that the possibility to perform such an {\em a posteriori} optimization is not unique to our witnesses. Rather, by introducing some auxiliary free parameters, one can, in principle, always optimize the choice of the witness depending on the measured data, as we illustrate in Sec.~\ref{Sec:Exp} (see also Appendix~\ref{App:Depth}).

Evidently, from the measurement of the local observables considered, it is possible to evaluate a many other expectation values (including those involving only a subset of parties) that we have not considered. The challenge then is to determine the $m$-separable bound, or the $k$-producible bound of the corresponding witness operator. Our work can thus  be seen as one of the first steps towards this general problem of finding the optimal linear entanglement witness directly from the measurement results. Even then, a linear entanglement witness generically works well only for a specific target state, or for quantum states that do not differ too much from it. Another line of research thus consists of employing a {\em nonlinear} entanglement witness for the detection of entanglement structure. Solving any of these problems in full generality is nonetheless clearly beyond the scope of the current research.

Experimentally, we have demonstrated how the entanglement structure of the tensor products of GHZ-type states can be inferred---with the help of some auxiliary assumptions---by systematically combining the results obtained during the measurement of our witnesses. More precisely, we have shown that the minimal entanglement structure deduced from these experimental results are exactly the entanglement structure that we expected from our experimental setup. The usefulness of the algorithmic procedures that we have introduced here in a more general setting, of course, remains to be investigated.

Finally, it is worth noting that the entanglement intactness witnesses introduced in the current work have very recently been generalized~\cite{Qi:Unpublished2} to the case of 1-dimensional cluster states. Since these states and GHZ states are both specific cases of a graph state, an open question that follows is whether these witnesses can be further generalized to cover a general graph state while maintaining their appealing feature of involving only two local observables (see also Ref.~\cite{Toth2005} in this regard). Given the importance of such states for one-way quantum computation~\cite{raussendorf2001one}, such a generalization may then be used to benchmark our progress towards the ultimate goal of demonstrating quantum supremacy.

\section{Acknowledgements}

We thank Denis Rosset for sharing his software, which facilitated our verification of the $k$-producible bounds presented in this work. We also thank two anonymous referees of PRX for the many helpful comments and suggestions. This work was supported by the National Natural Science Foundation of China (under Grant No.11404318, 11374284, 11674193 and 11425417), the National Fundamental Research Program (under Grant No. 2013CB922001), and the Chinese Academy of Sciences. Y.-C.L and J.-C.H were supported by the Ministry of Science and Technology, Taiwan (Grant No. 104-2112-M-006-021-MY3). H.L., Q. Z. and Z.-D. L. contributed equally to this work.

\appendix

\section{Theory}\label{theoryappendix}
The structure of multipartite entanglement is much richer than the bipartite case. An $n$-partite pure state $\ket{\phi}$ is said to be $m$-separable ($2\le m\le n$) if the $n$ parties can be divided into $m$ disjoint subsets $\{\mathcal{G}_i\}_{i=1,\ldots,m}$ such that $\ket{\phi}$ is the tensor product of a pure state $\ket{\psi_{\mathcal{G}_i}}$ from each of these subsets, i.e.,
\begin{equation}\label{ksepara}
\begin{aligned}
\ket{\phi}=\bigotimes_{i=1}^{m} \ket{\psi_{\mathcal{G}_i}}.
\end{aligned}
\end{equation}
The $m$-separability of a quantum state describes the extent of segregation. The larger the value of $m$, the more segregated $\ket{\phi}$ is. If $m=n$, we refer to it as a fully separable state.
Conversely, a non-$m$-separable state implies that it cannot be generated by segregating the subsystems into $m$ disjoint subsets and allowing arbitrary local manipulations within each subset.

Though (non)-$m$-separability of $\ket{\phi}$ already provides us with important information about the entanglement structure of $\ket{\phi}$, it is still not enough. How many parties each disjoint subset contains is also part of the specification of its entanglement structure. Let us denote by $n_i$  the number of subsystems involved in the subset $\mathcal{G}_i$ (note that $\sum_{i=1}^m n_i= n$); then $\ket{\phi}$ is said to be $k$-producible if the largest constituent of $\ket{\phi}$ involves at most $k$ parties, i.e., if $\max_i n_i = k$. Generating a $k$-producible state requires at most $k$-body entanglement and on the other hand, more than $k$-body entanglement is required in generating a not $k$-producible state.

For mixed state $\rho$, the $m$-separability and $k$-producibility of a general mixed state $\rho$ can be defined analogously: $\rho$ is $m$-separable (or $k$-producible) if it admits a convex decomposition in terms of $m$-separable ($k$-producible) pure states. If a quantum state $\rho$ is $k$-producible but not $(k-1)$-producible, we say it has an entanglement depth of $k$ \cite{Sorensen2001}. On the other hand, we say that a quantum state $\rho$ has an entanglement intactness of $m$ if it is $m$-separable but not $(m+1)$-separable. A genuinely $n$-partite entangled has an entanglement intactness (depth) of 1 ($n$) whereas a fully separable $n$-partite state has an entanglement intactness (depth) of $n$ (1). In particular, any quantum state that has an entanglement depth greater than 2 is conventionally said to contain multipartite (many-body) entanglement.

\subsection{A family of witnesses for non-$m$-separability with two local measurement settings}
\label{appendixseparablity}

In this section, we introduce the following two-parameter family of two-observable witnesses:
\begin{widetext}
\begin{equation}\label{App:eq:separability}
	\W^n_{se}(\alpha)=\alpha \Mz+\Mx \stackrel{m-\text{sep.}}{\le} \id_n\, \max\left\{\alpha, \frac{\alpha}{2^{m-1}}+1\right\}, \quad  \alpha\in(0, 2],
\end{equation}
\end{widetext}
where $\Mz=\left(\proj{0}\right)^{\otimes{n}}+\left(\proj{1}\right)^{\otimes{n}}$, $\Mx=\sigma_x^{\otimes{n}}$,
$\sigma_x$ is the Pauli $x$ matrix, $\{\ket{0},\ket{1}\}$ are the computational basis states, $\id_n$ is the $2^n\times 2^n$ identity matrix, and $m$-sep. in Eq.~\eqref{App:eq:separability} signifies that the inequality holds true at the level of the expectation value for all $m$-separable $n$-qubit states. In other words, for an arbitrary $n$-partite state $\rho$, if $\langle \W^n_{se}(\alpha) \rangle_\rho> \max\{\alpha, \frac{\alpha}{2^{m-1}}+1\}$, one certifies that $\rho$ has an entanglement intactness of $m-1$ or lower. Here, $\alpha$ is a free positive parameter that may be varied to identify the best possible upper {bound} on the entanglement intactness of $\rho$.

\subsubsection{Family of genuine $n$-qubit entanglement witnesses}

For the specific case of $m=2$, the witness of Eq.~\eqref{App:eq:separability} reduces to one that can be used to certify genuine $n$-qubit entanglement.

\begin{theorem}\label{genuine theorem}
Let $\rho$ be an arbitrary $n$-qubit biseparable state; then its expectation value for $\Mz$ and $\Mx$ [and hence $\ExpVal{\W^n_{se}(\alpha)}]$
satisfy
\begin{equation}\label{App:GME:W}
\!\!\!\alpha \ExpVal{\Mz}_\rho{\pm}\ExpVal{\Mx}_\rho \stackrel{\text{\rm 2-sep.}}{\le} \frac{\alpha}{2}+1, \quad  \alpha\in(0, 2].
\end{equation}
\end{theorem}

To prove the theorem, let us denote by $\vec{n}=\{n_1\}\{n_2\}$ a partition of the $n$ parties into a subset of $n_1$ parties and the complementary subset of $n_2=n-n_1$ parties. The invariance of $\W^n_{se}(\alpha)$ with respect to an arbitrary permutation of subsystem Hilbert spaces implies that in determining the biseparable bound, i.e., the maximal quantum value of $\ExpVal{\W^n_{se}(\alpha)}$ over all biseparable $n$-qubit states, the actual members of each subset $\mathcal{G}_1$ and $\mathcal{G}_2$ are irrelevant.

Without loss of generality, let us thus imagine that the first $n_1$ parties belong to $\mathcal{G}_1$, and denote by $S_{\vec{n}}$  the set of all $n$-qubit pure states that are biseparable with respect to this partitioning specified by $\vec{n}$. We may then write the biseparable bound, i.e., the maximal value of the right-hand-side of Eq.~\eqref{App:GME:W} as:
\begin{equation}\label{Eq:Dfn:Bisep}
	\max_{\text{\rm bisep. } \rho} \ExpVal{\W^n_{se}(\alpha)}_\rho =\max \limits_{\vec{n}}f_{\vec{n}} = \max \limits_{\{n_1\}\{n_2\}}f_{\{n_1\}\{n_2\}},
\end{equation}
where
\begin{equation}
\begin{aligned}\label{Eq:Dfn:fn}
f_{\vec{n}}:=\max \limits_{\ket{\phi} \in S_{\vec{n}} } \ExpVal{\W^n_{se}(\alpha)}_{\ket{\phi}}=\max \limits_{\ket{\phi} \in S_{\vec{n}} }  \mathrm{Tr}\left[(\alpha \Mz+\Mx)\proj{\phi}\right].
\end{aligned}
\end{equation}
As noted above, $\W^n_{se}(\alpha)$ is permutational invariant, we thus have
\begin{equation}\label{Eq:f:symmetry}
	f_{\{n_1\}\{n_2\}}=f_{\{n_2\}\{n_1\}}
\end{equation}

Next, we present a key observation that allows one to simplify the maximization of Eq.~\eqref{Eq:Dfn:Bisep} over an arbitrary $(n_1+n_2)$-qubit biseparable pure state to a maximization over an arbitrary $(n_1+1)$-qubit biseparable pure state.

\begin{lemma}\label{3}
The value of $f_{\vec{n}}$ for the bipartition of $n=n_1+n_2$ parties specified by $\vec{n}=\{n_1\}\{n_2\}$  is identical to the value of $f_{\vec{n'}}$ for the bipartition of $n'=n_1+1$ parties into $\vec{n'}=\{n_1\}\{1\}$.
\end{lemma}
\begin{proof}
Let $\ket{\phi_a}$ and $\ket{\phi_b}$, respectively, be arbitrary $n_1$-qubit and $n_2$-qubit pure states. From the definition of $f_{\vec{n}}$ given in Eq.~\eqref{Eq:Dfn:fn}, one finds that
\begin{equation}\label{Eq:supfab}
\begin{aligned}
&f_{\{n_1\}\{n_2\}} \\
&= \max \limits_{\ket{\phi_a},\ket{\phi_b} }  \mathrm{Tr}[(\alpha \Mz+\Mx)\proj{\phi_a}\otimes\proj{\phi_b}],\\
&=\max \limits_{\ket{\phi_a}}\,\, \mathrm{eig}_{\max}\{\mathrm{Tr}_a[(\alpha \Mz+\Mx)\proj{\phi_a}\otimes\id_{n_2}\}\\
&=\max \limits_{\ket{\phi_a}}\,\, \mathrm{eig}_{\max} M_b,\\
\end{aligned}
\end{equation}
where $M_b$ is an observable defined on the remaining $n_2$-qubit space,
\begin{equation}\label{Eq:M_b}
\begin{aligned}
&M_{b}=\alpha x(\proj{0})^{\otimes{n_2}}+\alpha y(\proj{1})^{\otimes{n_2}}+z(\sigma_x)^{\otimes{n_2}},
\end{aligned}
\end{equation}
and it depends on $\ket{\phi_a}$ via $ x= \bra{\phi_a} (\proj{0})^{\otimes{n_1}}\ket{\phi_a}$,
$ y= \bra{\phi_a} (\proj{1})^{\otimes{n_1}}\ket{\phi_a}$, and
$ z= \bra{\phi_a} (\sigma_x)^{\otimes{n_1}} \ket{\phi_a}$.

For any integer $n_2\ge 1$, $M_b$ has the following generic (sparse) matrix representation:
\begin{equation}\label{Eq:Mrest}
M_b=\left(   \begin{array}{ccc}
  \alpha x &  & z\\
     & \iddots  &       \\
   z& &\alpha y\\
    \end{array}
\right).
\end{equation}
Moreover, it can be verified that $M_b$ has two $(2^{n_2-1}-1)$-fold degenerate eigenvalues $\pm z$ and two nondegenerate eigenvalues  $\lambda_{\pm}=\tfrac{\alpha(x+y)}{2}\pm\sqrt{z^2+\tfrac{\alpha^2(x-y)^2}{4}}$. Since $x,y$ are non-negative, it follows that $|z| \le \frac{\alpha(x+y)}{2}+\sqrt{z^2+\frac{\alpha^2(x-y)^2}{2}}$. Hence, the largest eigenvalue of $M_b$ is necessarily  $\lambda_\text{\tiny max}=\lambda_+$.

Importantly, as long as $n_2\ge 1$, the same conclusion holds regardless of the actual number of qubits involved in the definition of $M_b$. In other words, while the size of $M_b$ depend on $n_2$, its largest eigenvalue $\lambda_\text{\tiny max}$, and hence $f_{\vec{n}}$ only depends on $\ket{\phi_a}$ via $x,y$ and $z$. Consequently, the very same argument can be repeated in the computation of $f_{\vec{n'}}$ with $\vec{n}'=\{n_1\}\{1\}$ to arrive at the conclusion that the largest eigenvalue of the matrix corresponding to Eq.~\eqref{Eq:M_b} is again $\lambda_+$. Therefore, $f_{\vec{n}}$ for the bipartition specified by $\vec{n}=\{n_1\}\{n_2\}$  coincides with  $f_{\vec{n'}}$ for the bipartition specified by $\vec{n'}=\{n_1\}\{1\}$.
\end{proof}

Now, we are in the position to prove Theorem~\ref{genuine theorem} by combining Lemma~\ref{3}, Eq.~\eqref{Eq:f:symmetry} and explicitly calculating the maximal eigenvalue of the resulting $2\times2$ matrix.
\begin{proof}
From Lemma~\ref{3}, we note that for arbitrary $n_1,n_2$ such that $n=n_1+n_2$, we have $f_{\vec{n}}=f_{\{n_1\}\{n_2\}}=f_{\{n_1\}\{1\}}$. Using Eq.~\eqref{Eq:f:symmetry}, $f_{\{n_1\}\{1\}}$ can be rewritten as $f_{\{1\}\{n_1\}}$. Applying Lemma~\ref{3} again to $f_{\{1\}\{n_1\}}$, we thus find that $\max  \limits_{\text{\rm bisep. } \rho} \ExpVal{\W^n_{se}(\alpha)}_\rho =\max \limits_{\vec{n}}f_{\vec{n}}=f_{\{1\}\{1\}}$. Computation of the biseparable bound thus amounts to computing the maximal eigenvalue of $M_b=\left(   \begin{array}{cc}
  \alpha x &  z\\
   z&\alpha y\\
    \end{array}
\right)$. Let us adopt the  parameterization $\ket{\phi_a}=\cos\theta\ket{0} + {\rm e}^{{\rm i}\varphi}\sin\theta\ket{1}$, then $x=\cos^2\theta$, $ y=\sin^2\theta$, $z=\cos\varphi\sin2\theta$ and $f_{\{1\}\{1\}}=\tfrac{\alpha}{2} + \tfrac{1}{2}\sqrt{\alpha^2\cos^22\theta+4\sin^22\theta\cos^2\varphi}$. For $\alpha\in(0,2]$, the term in the square root is clearly maximized by setting $\theta=\tfrac{\pi}{4}$, $\varphi=0$, thereby giving a biseparable bound of $f_{\{1\}\{1\}}=\tfrac{\alpha}{2} +1$.
\end{proof}

We thus prove the result of Theorem~\ref{genuine theorem}. Note that the witnesses of Ref.~\cite{Toth2005} is a special case of our witnesses corresponding to $\alpha=2$.

\subsubsection{Family of (non)-$m$-separability  witnesses}

For the more general family of witnesses for detecting non-$m$-separability (and hence an entanglement intactness of $m-1$ or lower), we follow a very similar procedure as that adopted in the last section. Specifically, we first iteratively apply Lemma~\ref{3} and Eq.~\eqref{Eq:f:symmetry} to show that determining the $m$-separable bound amounts to computing $f_{\vec{n}}$ where
\begin{equation}
	\vec{n}=\stackrel{m\,\,\text{times of}\,\,\{1\}}{\overbrace{\{1\}\cdots\{1\}}}.
\end{equation}
Next, if we adopt the generic parametrization of setting $\ket{\phi_i} = \cos\theta_i\ket{0}+{\rm e}^{{\rm i}\varphi_i}\sin\theta_i\ket{1}$, then the computation of  $f_{\vec{n}}$ is equivalent to maximizing the largest eigenvalue of the qubit observable $M_b$, i.e.,
\begin{equation}\label{Eq:Dfn:f1-n}
\begin{aligned}
f_{\{1\}^{\times m}}=\tfrac{1}{2}( \alpha(x+y)+\sqrt{\alpha^2(x-y)^2+4z^2})\\
\end{aligned}
\end{equation}
where
\begin{equation}\label{Eq:xyz}
	x=\prod_{i=1}^{m-1}\cos^2\theta_i,\,\,  y=\prod_{i=1}^{m-1}\sin^2\theta_i, \,\, z=\prod_{i=1}^{m-1}\cos\varphi_i\sin2\theta_i.
\end{equation}
As is evident in Eq.~\eqref{Eq:Dfn:f1-n}, we may, without loss of generality, set $\varphi_i=0$ for all $i$ in our maximization of $f_{\{1\}^{\times m}}$.

When $\alpha \ge  \frac{2^{m-1}}{2^{m-1}-1}$, it can be shown that
\begin{equation}
\begin{aligned}\label{kseparablebound}
f_{\{1\}^{\times m}} \le \alpha, 
\end{aligned}
\end{equation}
whereas for $0<\alpha <  \frac{2^{m-1}}{2^{m-1}-1}$, one has
\begin{equation}
\begin{aligned}\label{kseparablebound2}
f_{\{1\}^{\times m}}
&\le \tfrac{\alpha}{2^{m-1}}+1.
\end{aligned}
\end{equation}
Consequently, the $m$-separable bound is upper bounded by $f_{\{1\}^{\times m}}\le \max\{\alpha,\tfrac{\alpha}{2^{m-1}}+1\}$.

To see that inequality~\eqref{kseparablebound} holds, note from Eq.~\eqref{Eq:Dfn:f1-n} that this inequality is equivalent to:
\begin{equation}
\begin{split}
&\tfrac{1}{2}( \alpha(x+y)+\sqrt{\alpha^2(x-y)^2+4z^2})\le \alpha,\\
\Longleftrightarrow\quad &z^2\le \alpha^2(1-x)(1-y).
\end{split}
\end{equation}

From the non-negativity of $1-x$, $1-y$, the relationship between $x,y$, $z$ given in Eq.~\eqref{Eq:xyz}, and the assumption that $\alpha \ge  \frac{2^{m-1}}{2^{m-1}-1}$, we see that proving
Eq.~\eqref{kseparablebound} for this interval of $\alpha$ amounts to proving
\begin{equation}
\begin{split}
&z^2 \le \left(\frac{2^{m-1}}{2^{m-1}-1}\right)^2(1-x)(1-y),\\
\Longleftrightarrow\quad &2^{m-1}(2^{m-1}-2)xy \le 1-x-y.
\end{split}
\end{equation}

For the convenience of subsequent discussions, let us define
\begin{equation}\label{Eq:xiyi}
	x_i:=\cos^2\theta_i, \quad y_i:=\sin^2\theta_i.
\end{equation}
Using the mathematical identity $  \prod_{i=1}^{m-1} (x_i+y_i)=1$, we can now make both sides of the above inequality a degree $2(m-1)$ homogeneous polynomial in the variables $\{x_i,y_i\}_{i=1}^{m-1}$, namely,
\begin{widetext}
\begin{equation}\label{Eq:Homo}
\begin{aligned}
2^{m-1}(2^{m-1}-2)\prod_{i=1}^{m-1} x_i y_i \le \prod_{i=1}^{m-1}(x_i +y_i)\left[\prod_{j=1}^{m-1}(x_j +y_j)-\prod_{k=1}^{m-1}x_k-\prod_{k=1}^{m-1}y_k\right].
\end{aligned}
\end{equation}
\end{widetext}
Evidently, the polynomial on the RHS of inequality~\eqref{Eq:Homo} consists of $2^{m-1}(2^{m-1}-2)$ monomials, each of degree $2(m-1)$, while the left-hand-side (LHS) consists of $2^{m-1}(2^{m-1}-2)$ times the {\em same} monomial. The key observation leading to the bound given in Eq.~\eqref{kseparablebound} is that when a complementary pair of monomials from the RHS are combined with two of the  monomials from the LHS, one obtains a square of some polynomial. Consequently, after subtracting the LHS from the RHS of Eq.~\eqref{Eq:Homo}, we end up with a sum of squares (SOS) of polynomials, which are necessarily non-negative, thereby showing that the RHS is greater than or equal to the LHS.

To this end, let $\N=\{1,2,\ldots,m-1\}$ denote the set of indices ranging from 1 to $m-1$. Then, it is not difficult to see that all monomials appearing in Eq.~\eqref{Eq:Homo}
take the form
\begin{equation}\label{Eq:g}
	g=\prod_{i\in H} x_iy_i\prod_{j\in \mH^\complement}\xi_j^2
\end{equation}
where $\xi_i$ either equals $x_i$ or $y_i$ for each $i$, $\mH$ is a subset of $\N$ such that for all $i\in\mH$, $g$ is linear in both $x_i$ and $y_i$, while $H^\complement$ is the complement of $\mH$ in $\N$, i.e., the subset of $\N$ such that $g$ is either quadratic in $x_i$ or $y_i$. For example, when using Eq.~\eqref{Eq:g} to express the monomials appearing in the LHS of Eq.~\eqref{Eq:Homo}, we have $H=\N$, or equivalently $H^\complement$ being the empty set.

Let us further define the monomial complementary to $g$ as $\bar{g}:= \prod_{i\in \mH} x_iy_i\prod_{j\in H^\complement}^{m-1}\bar{\xi}_j^2$, where $\bar{x}_i=y_i$ and $\bar{y}_i=x_i$; i.e., $\bar{g}$ is obtained from $g$ by changing each $x_i$ to $y_i$ and vice versa. Subtracting from $g$ any of the monomials appearing in the LHS of Eq.~\eqref{Eq:Homo} gives
$	g - \prod_{i=1}^{m-1} x_i y_i
	=\prod_{i\in H} x_iy_i\prod_{j\in \mH^\complement}\xi_j (\xi_j-\bar{\xi}_j)$.
Similarly, subtracting from $\bar{g}$ any of the monomials appearing in the LHS of Eq.~\eqref{Eq:Homo} gives
	$\bar{g} - \prod_{i=1}^{m-1} x_i y_i=-\prod_{i\in H} x_iy_i\prod_{j\in \mH^\complement}\bar{\xi}_j (\xi_j-\bar{\xi}_j)$.
Combining these expressions
while recalling from Eq.~\eqref{Eq:xiyi} the non-negativity of $x_i,y_i$, we then have
\begin{equation}\label{Eq:Combined}
	g + \bar{g} - 2\prod_{i=1}^{m-1} x_i y_i = \prod_{i\in H} x_iy_i\prod_{j\in \mH^\complement} (\xi_j-\bar{\xi}_j)^2\ge 0,
\end{equation}
where the non-negativity of the overall expression follows from it being the square of some polynomial. To complete the proof, it suffices to note that for all $g$ appearing in the RHS of Eq.~\eqref{Eq:Homo}, $\bar{g}$ also appears on the RHS as one of the $2^{m-1}(2^{m-1}-2)$ monomials. Thus, it follows from Eq.~\eqref{Eq:Combined} that the RHS-LHS of Eq.~\eqref{Eq:Homo} is indeed a SOS and thereby shows the validity of inequality~\eqref{Eq:Homo}, as well as that of Eq.~\eqref{kseparablebound}.
The proof of Eq.~\eqref{kseparablebound2} proceeds analogously to that given above.

Finally, to see that the $m$-separable bound given in Eq.~\eqref{kseparablebound} is tight, it suffices to note that inequality~\eqref{kseparablebound} is saturated when $\sin^2\theta_i=1, \cos^2\theta_i=0 $ (or $\sin^2\theta_i=0, \cos^2\theta_i=1 $) for every $i$. Similarly, the $m$-separable bound given in Eq.~\eqref{kseparablebound2} is tight as the corresponding bound is saturated when  $\sin^2\theta_i=\cos^2\theta_i=1/2 $ for  every $i$.

Note that the non-$m$-separability of a state also gives nontrivial information about its entanglement depth. For example, if a state $\rho$ is 3-separable, then $\rho$ is at least $\lceil\tfrac{n}{3}\rceil$-producible. Likewise, an $m$-separable state is at least $\lceil\tfrac{n}{m}\rceil$-producible. Thus, if the measured value of $s=\ExpVal{\Mz+\Mx}$ for a state is such that $\frac{3}{2}\ge s> \frac{5}{4}$, then the measured state is not 3-separable, but is biseparable; its entanglement depth is thus at least $\lceil\frac{n}{2}\rceil$.

\subsubsection{{White-noise robustness of the non-$m$-separability witnesses for GHZ states}}

In general, one may hope to apply our non-$m$-separability witnesses of Eq.~\eqref{App:eq:separability} to deduce some nontrivial lower bound on the entanglement intactness of any $n$-qubit (entangled) state prepared in the laboratory. In practice, however, as with any other entanglement witnesses, they are not without their limitations. For example, it is easy to verify that the entanglement present in the $n$-qubit $W$ state cannot be certified at all by an evaluation of our witness $\ExpVal{\W^n_{se}(\alpha)}$. Rather, our witness $\ExpVal{\W^n_{se}(\alpha)}$ seems to be better suited for certain classes of states, such as the $n$-partite GHZ state $\ket{\textrm{GHZ}_n}=\frac{1}{\sqrt{2}}(\ket{0}^{\otimes n}+\ket{1}^{\otimes n})$, and their generalization (more on this below). In fact, the proposed witness with any parameter $\alpha\in (0,2]$ can detect the genuine $n$-partite entanglement present in these states.

To see how our witnesses fare in the presence of (white) noise, let us consider the mixed state\footnote{The performance of these witnesses for some other experimentally inspired noise model can be found in Appendix~\ref{noiseappendix}. }
\begin{equation}\label{GHZmixed}
\begin{aligned}
\rho=(1-p_\text{noise}){\proj{\textrm{GHZ}_n}}+p_\text{noise}{\frac{\mathbb{I}_n}{2^n}}.
\end{aligned}
\end{equation}
Evaluating $\ExpVal{\W^n_{se}(\alpha)}$ against this state gives
\begin{equation}\label{Eq:GHZ:ExpVal}
	\alpha \ExpVal{\Mz}+\ExpVal{\Mx}=(\alpha+1)(1-p_\text{noise})+\alpha \frac{p_\text{noise}}{2^{n-1}}.
\end{equation}
Comparing this with the maximal value of $\ExpVal{\W^n_{se}(\alpha)}$ attainable by a {bi}separable state, i.e., $\frac{\alpha}{2} +1$, leads to the identification of
\begin{equation}\label{Eq:GHZ:NoiseRange}
	p_\text{noise}<\frac{\alpha}{2+(2-2^{2-n})\alpha}
\end{equation}
as a genuine multipartite-entangled $\rho$. Evidently, to maximize the right-hand-side of Eq.~\eqref{Eq:GHZ:NoiseRange}, the optimal choice of $\alpha\in(0,2]$ is given by $\alpha=2$. This leads to
\begin{equation}
\begin{aligned}
	p_\text{noise}<\frac{1}{3-2^{2-n}},
\end{aligned}
\end{equation}
which tends to $\frac{1}{3}$, i.e., less than $\frac{1}{2}$---the maximal noise tolerance achievable with the more well-known witness tailored for the GHZ state, $W=\frac{1}{2}-{\proj{\textrm{GHZ}_n}}$.

When the noise parameter $p_\text{noise}$ of $\rho$ in Eq.~\eqref{GHZmixed} increases, the state becomes more segregated, thus showing larger values of entanglement intactness.  For the detection of the non-$m$-separability of these states, we compare instead the expectation value of Eq.~\eqref{Eq:GHZ:ExpVal} with the $m$-separable bound of $\max\{\frac{\alpha}{2^{m-1}}+1,\alpha\}$. Because of the linearity of these expressions in $\alpha$, the optimal choice of $\alpha$ takes place when they are equal, i.e., when $\alpha=\frac{1}{1-2^{1-m}}$, thereby giving an $m$-separable bound of $\frac{2^{m-1}}{2^{m-1}-1}$. Solving for the corresponding threshold noise parameter shows that our witnesses reveal an upper bound on the entanglement intactness of $m-1$ for $\rho$ whenever
\begin{equation}
	p_\text{noise}<\frac{2^m-2}{2 (2^m - 2^{m - n}-1)}.
\end{equation}

\subsubsection{{White-noise robustness of the non-$m$-separability witnesses for GHZ-like states}}\label{App:modification}

Having understood how our witnesses for non-$m$-separability work for GHZ states and their mixture with white noise, we now perform a similar analysis for the generalized GHZ state involving an arbitrary coherent superposition between $\ket{0}^{\otimes n}$ and $\ket{1}^{\otimes n}$. Specifically, let $\ket{\textrm{GHZ}_n(\theta,\phi)}:=\cos\theta\ket{0}^{\otimes n}+{\rm e}^{{\rm i}\phi}\sin\theta\ket{1}^{\otimes n}$ where $\theta\in(0,\tfrac{\pi}{4}]$ while $\phi\in(0,2\pi]$. For simplicity, our discussion here will focus mainly on the detection of genuine $n$-partite entanglement present in (the noisy version of) such states.

To this end, note that for $\ket{\textrm{GHZ}_n(\theta,\phi)}$, the linear combination of expectation values [appearing in $\ExpVal{\W^n_{se}(\alpha)}$ and $\ExpVal{\W^{n'}_{se}(\alpha)}$, see the last paragraph of Sec.~\ref{Structure}] gives
\begin{equation}
	\alpha \ExpVal{\Mz} \pm \ExpVal{\Mx}=\alpha\pm\sin2\theta\cos\phi.
\end{equation}
Clearly, if we take $\alpha=2$, then independent of the value of $\theta\in(0,\tfrac{\pi}{4}]$ and except when $\phi=\tfrac{\pi}{2},\tfrac{3\pi}{2}$, the right-hand-side of the above expression---after maximizing over both signs $\pm$---always exceeds the biseparable bound of $\alpha=2$. Thus, the genuine multipartite entanglement of {\em almost all} $\ket{\textrm{GHZ}_n(\theta,\phi)}$, except when $\phi=\tfrac{\pi}{2},\tfrac{3\pi}{2}$, can be certified via our nonseparability witnesses given in Eq.~\eqref{App:eq:separability}. In fact, even the genuine multipartite entanglement present in the two remaining cases can be taken care of analogously by applying an appropriate local unitary transformation to $\W^n_{se}(\alpha)$. Specifically, for $\phi=\tfrac{\pi}{2},\tfrac{3\pi}{2}$, it suffices to apply the unitary transformation $\id_2^{\otimes n-1}\otimes\left(\begin{smallmatrix} 1 & 0 \\ 0 & \pm i\end{smallmatrix}\right)$ to $\W^n_{se}(\alpha)$, an evaluation of the resulting witness for $\ket{\textrm{GHZ}_n(\theta,\phi=\tfrac{\pi}{2})}$ then gives $\alpha+\sin2\theta$, which always exceeds the biseparable bound of $\alpha$.

To determine the white-noise robustness of our witnesses against these generalized GHZ states, we consider
\begin{equation}
\begin{aligned}\label{Eq:rhomixed:generalizedGHZ}
	\rho=(1-p_\text{noise})\proj{\textrm{GHZ}_n(\theta,\phi)}+p_\text{noise}{\frac{\mathbb{I}_n}{2^n}}.
\end{aligned}
\end{equation}
By a calculation similar to that presented in the last section [but now considering both $\ExpVal{\W^n_{se}(\alpha)}$ and $\ExpVal{\W^{n'}_{se}(\alpha)}$], one finds that the genuine $n$-partite entanglement present in these noisy versions of $\ket{\textrm{GHZ}_n(\theta,\phi)}$ can always be certified as long as
\begin{equation}
	p_\text{noise}<\frac{\sin2\theta|\cos\phi|}{2+\sin2\theta|\cos\phi|-2^{2-n}}.
\end{equation}
Likewise, it can be shown that as long as
\begin{equation}
	p_\text{noise}<\frac{2^n(2^m-2)\sin2\theta|\cos\phi|}{2^n(2^m-2)\sin2\theta|\cos\phi|+2^m(2^n-2)},
\end{equation}
one could certify that the $\rho$ given in Eq.~\eqref{Eq:rhomixed:generalizedGHZ} has an entanglement intactness upper bounded by $m-1$.

\subsection{A family of witnesses for non-$k$-producibility with two local measurement settings}
\label{App:Depth}

Our witnesses for entanglement depth have their origin in the family of device-independent (DI) witnesses for entanglement depth given in Ref.~\cite{PhysRevLett.114.190401}:
\begin{equation}\label{Eq:EDWitnesses}
	\I_n^k(\gamma): \frac{\gamma}{2^n}\sum_{\vec{x}\in\{0,1\}^n} E_n(\vec{x})-E_n(\vec{1}_n)\,  \stackrel{\stackrel{\mbox{\tiny $k$-producible}}{\mbox{\tiny states}}}{\le} \SqkM{k}(\gamma).
\end{equation}
where $\vec{x}=(x_1,x_2,\ldots,x_n)$ is an $n$-component vector describing the combination of measurement settings, $x_i\in\{0,1\}$, $E_n(\vec{x})$ is the $n$-partite full correlator (the expectation value of an $n$-partite $\pm1$-valued outcome observable), and $\SqkM{k}(\gamma)$ is the maximal quantum value of $\I_n^k(\gamma)$ attainable by any $k$-producible state. Some explicit values of these DI $k$-producible bounds (which hold for arbitrary dimensional $k$-producible states and arbitrary local $\pm1$-valued observables) for the case of $\gamma=2$ are~\cite{PhysRevLett.114.190401}: $\SqkM{1}=1$, $\SqkM{2}=\sqrt{2}$, $\SqkM{3}=\tfrac{5}{3}$, $\SqkM{4}=1.8428$ etc.

For an $n$-partite GHZ state, a good choice of local observables inspired by those of Ref.~\cite{PhysRevLett.114.190401} is given by setting $\A_{\pm}= \cos\theta_\pm \sigma_x +\sin\theta_\pm\sigma_y$, where $\theta_\pm\in\mathbb{R}$, while $\A_-$ and $\A_+$ are, respectively, the local observables for $x_i=0$ and $x_i=1$ (for all $i$). In particular, for the 8-partite states that we managed to produce experimentally,  based on the measured values of $\ExpVal{\mathcal M_Z}$ and $\ExpVal{\mathcal M_X}$, our offline numerical optimizations suggest that $\theta_\pm=\tfrac{3(1\pm 8)}{80}$ is a reasonably robust choice for witnessing the entanglement depth. Substituting these into the left-hand-side of Eq.~\eqref{Eq:EDWitnesses} and denoting the global Hermitian observables as $\W^n_{de}$, i.e., $\I_n^k=\tr(\rho\,\W^n_{de})$, we then see that:
\begin{equation}\label{DDInk}
	\W^n_{de}(\gamma)= \gamma\left(\frac{\A_-+\A_+}{2}\right)^{\otimes n}-\A_+^{\otimes n}=\gamma\kappa^n\A-\A_+^{\otimes n},
\end{equation}
where $\A = (\frac{\A_-+\A_+}{2\kappa})^{\otimes n}$ is a $\pm1$-valued Hermitian observable and $\kappa=\cos\tfrac{3}{10}$ is a normalization constant.

As Eq.~\eqref{Eq:EDWitnesses} holds for an arbitrary choice of local observables, we thus see that $\ExpVal{\W^n_{de}(\gamma)}\stackrel{\stackrel{\mbox{\tiny $k$-producible}}{\mbox{\tiny states}}}{\le} \SqkM{k}(\gamma)$ already represents a family of witnesses for entanglement depth. For the specific choice of observables given above, however, these $k$-producible bounds can be considerably tightened via numerical optimizations.

Specifically, our goal is to compute
\begin{equation}\label{Eq:gamma-nk}
	\beta_{8,k}(\gamma)=\max \limits_{\text{$k$-prod.\ }\rho}\tr\left[\rho\, \W^{8}_{de}(\gamma) \right],
\end{equation}
i.e., to optimize the expectation value of $\W^{8}_{de}$ over all possible 8-qubit $k$-producible states. A few simplifications can immediately be made. First, since the objective function $\tr\left[\rho\, \W^{8}_{de}(\gamma) \right]$ is linear in $\rho$, there is no need to consider convex mixtures of $k$-producible 8-qubit states in the optimization. In other words, it suffices to consider $\rho=\otimes_{i=1}^m \rho_i$ where each $\rho_i=\proj{\psi_i}$ is at most $k$-partite. Second, as $\W^{8}_{de}$ is invariant under arbitrary permutation of parties, it suffices to consider one particular partitioning separating the 8 parties into $\lfloor \tfrac{8}{k}\rfloor$ groups of $k$ parties (possibly plus a remaining group of $8\,\mod\,k$ parties).

Even with these simplifications, there is no straightforward way to determine the values of $\beta_{8,k}$, as the characterization of separable states---and, more generally $k$-producible quantum states---is a computationally difficult problem. Instead, we  numerically determine some (matching) upper bound for the $k$-producible bound $\beta_{8, k}$ by employing (and generalizing) the idea of symmetric extension proposed in Ref.~\cite{doherty2004complete} to the present problem.

For example, in order to determine (an upper bound on) the 3-producible bound $\beta_{8, 3}$, it suffices to consider $\rho=\rho_A\otimes\rho_B\otimes\rho_C$ where both $\rho_A$ and $\rho_B$ are three-qubit states and $\rho_C$ is a two-qubit state. Clearly, for all such states, there exists an $(n_1,n_2,n_3)$-copy symmetric extension $\tilde{\rho}$ (e.g., {$\tilde{\rho}=\rho_A^{\otimes n_1}\otimes \rho_B^{\otimes n_2}\otimes \rho_C^{\otimes n_3}$}) such that $\pi \tilde{\rho}\pi =\tilde{\rho}$ and $\tr_{\{A^{\otimes n_1-1}B^{\otimes n_2-1}C^{\otimes n_3-1}\}}\tilde{\rho}=\rho$ where $\pi$ is the projector onto the symmetric subspace of $n_1$ copies of A's Hilbert space, $n_2$ copies of B's Hilbert space, and $n_3$ copies of C's Hilbert space  while $\tr_{\{A^{\otimes n_1-1}\}}\rho$ means a partial trace over $n_1-1$ copies of A's Hilbert space etc. Therefore, a legitimate upper bound on $\beta_{8, 3}$ can be obtained by solving the following semidefinite program:
\begin{equation}\label{Eq:PPTLSE}
\begin{split}
	\max\quad &\tr\left[\rho\, \W^n_{de}(\gamma) \right],\\
	\text{s.t. \ \ \ } & \rho\succeq 0,\quad \tr(\rho)=1,\\
	 & \tilde{\rho}\succeq 0,\quad \tr(\tilde{\rho})=1,\quad \pi\tilde{\rho}\pi=\tilde{\rho}\\
	 & \tilde{\rho}^\text{\tiny T$_j$}\succeq 0\quad\forall\quad j\in \I
\end{split}
\end{equation}
where $\mathcal{O}\succeq 0$ represents the positive-semidefinite requirement of $\mathcal{O}$, $\tilde{\rho}^\text{\tiny T$_j$}$ represents the partial transposition~\cite{peres1996separability} of $\tilde{\rho}$ with respect to subsystem $j$ and $\I$ is the set of indices representing {\em all} possible combinations of varying numbers of copies of A, B, and C's Hilbert spaces. Therefore, if $n_1=n_2=1$ and $n_3=2$, the last line of Eq.~\eqref{Eq:PPTLSE} represents the following set of constraints:
\begin{equation}
	\tilde{\rho}^\text{\tiny T$_A$},\quad \tilde{\rho}^\text{\tiny T$_B$},\quad \tilde{\rho}^\text{\tiny T$_{CC}$},\quad \tilde{\rho}^\text{\tiny T$_{C}$},\quad \tilde{\rho}^\text{\tiny T$_{AC}$},\quad \tilde{\rho}^\text{\tiny T$_{BC}$}.
\end{equation}
In this particular case, the upper bound on $\beta_{8,3}(\gamma=2)$ that we obtained is 1.1699.

In general, the upper bounds that we obtained by solving Eq.~\eqref{Eq:PPTLSE} are not necessarily tight, as the set of $\rho$ that we optimized over is generally a superset of the set of $k$-producible states. In our case, however, we could certify the tightness of these bounds by explicitly parametrizing a general 8-qubit $k$-producible pure state, a general dichotomic qubit observable (3 parameters for every such observable), and applying standard (but heuristic) algorithms to optimize Eq.~\eqref{Eq:gamma-nk} over all these parameters 1000 times for each value of $k$. Our results for these optimizations are summarized in Table~\ref{Tbl:DDBounds} and in Fig.~\ref{fig:k-prod}.

\begin{figure}[t!]
\centering
\includegraphics[width=0.9\columnwidth]{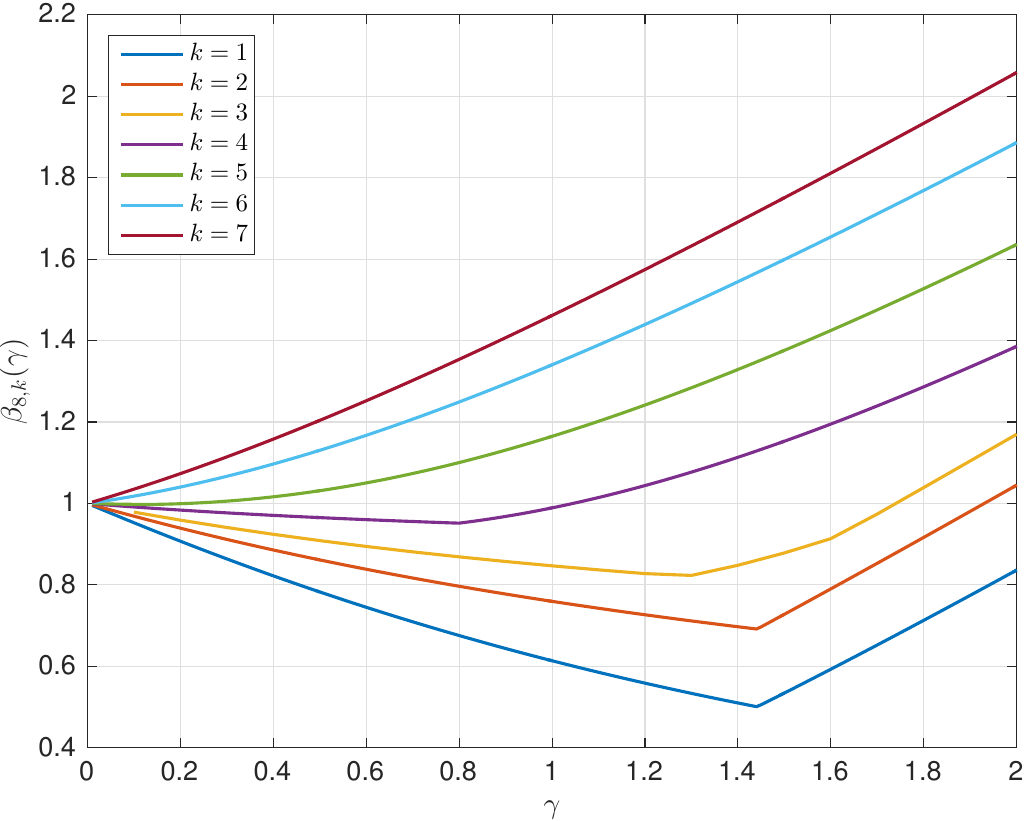}
\caption{\textbf{Numerically determined $k$-producible bounds {$\beta_{8,k}(\gamma)$} for $k=1,2,\ldots,7$ and for all $\gamma\in(0,2]$.}}
\label{fig:k-prod}
\end{figure}

\begin{table}[h!]
\centering
\begin{tabular}{|c|c|c|c|c|}
\hline
$k$ & $\gamma$ & {$\beta_{8,k}(\gamma)$} & Copies & Dimensions \\ \hline
1 & 2 & 0.8365  & $(1,1,2,2,2,2,2)$ & (2,\ldots,2)\\
2 & 2 &  1.0450  &  $(1,1,1,2)$ & (4,4,4,4)\\
2 & 1.6 &  0.7904  &  $(1,1,1,2)$ & (4,4,4,4)\\
3 & 2 &  1.1699  & (1,1,2) & (8,8,4)\\
3 & 1.6 &  0.9137  & (1,1,2) & (8,8,4)\\
4 & 2 &  1.3856  & (1,1) & (16,16) \\
5 & 2 &  1.6357  & (1,1) & (32,8)\\
6 & 2 &  1.8858  & (1,1) & (64,4)\\
7 & 2 &  2.0578  & (1,1) & (128,2) \\
\hline
\end{tabular}
\caption{\textbf{Summary of numerically determined $k$-producible bounds {$\beta_{8,k}(\gamma)$ }for $k=1,2,\ldots,7$ and some specific values of $\gamma$.} The last two columns give, respectively, the number of copies considered for each group and the Hilbert space dimension of each group in our computation of the (matching) upper bound on $\beta_{8,k}(\gamma)$.}\label{Tbl:DDBounds}
\end{table}

\setcounter{figure}{0}
\setcounter{table}{0}
\renewcommand{\theequation}{B\arabic{equation}}
\renewcommand{\thefigure}{B\arabic{figure}}
\renewcommand{\thetable}{B\arabic{table}}

\section{More experimental details and data processing}\label{experimentappendix}

\subsection{Entanglement preparation}\label{sourceappendix}

\begin{figure*}[t!]
\centering
\includegraphics[width=2\columnwidth]{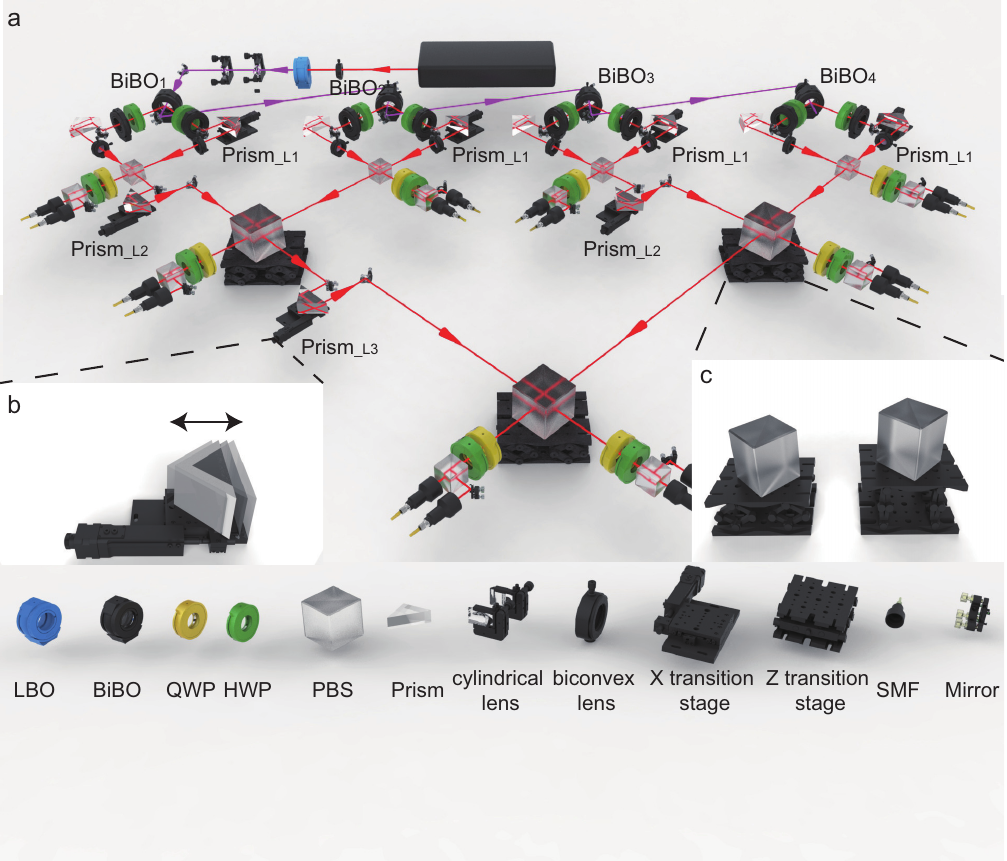}
\caption{\textbf{The detailed experimental setup.}}
\label{fig:setup_details}
\end{figure*}

The detailed experimental setup is shown in Fig.~\ref{fig:setup_details}. A femtosecond pulse, which is with {a} duration of 130 fs, central wavelength of 780 nm, repetition rate of 80 MHz and power intensity of 3.8 W, is focused to an LBO crystal with a waist of 50$\mu$m by a biconvex lens with focal length of 50 mm. As the instantaneous intensity of the focused pulse on LBO is extremely high, we add a Y-direction transition stage under the LBO, which moves 50~$\mu$m every 30 seconds to avoid destroying the LBO crystal. With such a device, we can observe a stable second harmonic generation (SHG) with an efficiency of 42.1\%, i.e., generating ultraviolet pulses with an average power of 1.6W. The optical mode of the generated ultraviolet pulse disperses differently in the x direction and the $y$ direction. To get a good Gaussian mode, we use two cylindrical lenses---one works for the $x$ direction and the other works for the $y$ direction---to reform the beam. After the reforming, the ultraviolet pulse is a Gaussian-like beam focused onto the first BiBO crystal with a beam waist of 170 $\mu$m. We use 0.6~mm BiB$_{3}$O$_{6}$ (BiBO) crystals cut at (111.4$ ^{\circ}$, 55.1 $^{\circ}$) to generate entangled-photon pairs \cite{chen2017observation}, in which case the two cones overlap along two lines separated by an angle of 6.9$^{\circ}$. Compared to the traditional Beta barium borate (BBO) crystals \cite{Kwiat95}, BiBO crystals with these cutting angles are expected to have a smaller spatial walk-off angle and a higher type-II second-order nonlinear coefficient \cite{Halevy11}. Then,  the entangled-photon pairs can be generated with a higher probability, and collected with a higher efficiency. The ultraviolet pulse is directed to the second BiBO crystal (BiBO$_2$), and refocused on the second BiBO$_2$ by a biconvex lens ($f=100$mm) with a beam waist of 170$\mu$m. The same operations are performed successively on the ultraviolet pulse to make it shine on BiBO$_3$ and BiBO$_4$ with the same beam waist of 170$\mu$m. With this choice of beam waist and under such a pumping condition, SPDC on BiBO crystal takes place with a considerable generation rate and good collective efficiency. During our experiment, to suppress the higher-order emission rate in SPDC, we attenuated the average power of ultraviolet pulse to 500 mW, under which we observed 3$\times 10^5$ two-fold coincidences in four entangled photon pairs. The average collection efficiency in modes 1-8 is 39\% with maximal 42\% and minimal 37\%. 

In our experiment, there are four active interferometers at minimum, and seven at most. To make sure that the interfered photons arrive at the polarization beamsplitter simultaneously, we add a motor-controlled prism in one arm of each interferometer [as shown in Fig.~\ref{fig:setup_details}(b)]. The stepping motor moves the prism with a minimum step size of 1$\mu$m, which is good enough to find the biggest visibility in the interference (the coherence length in our case is 200$\mu$m). The seven motor-controlled prisms need to be controlled by a strict sequence. We divide them into three layers: The first layer contains four prisms, which are embedded in the four interferometer generating entangled photon pairs; the second layer contains two prism which are employed in the interferometer to generate four-photon entanglement; the last layer contains one prism in the interferometer to generate eight-photon entanglement. There is no order in controlling the prisms in the same layer, but the controlling order between different layers must follow: layer 1$\rightarrow$ layer 2 $\rightarrow$ layer 3. When collecting data, all seven stepping motors need to be stable for  dozens of hours simultaneously. 

The three PBSs embedded in the interferometric geometry are attached in a lifting platform shown in Fig.~\ref{fig:setup_details}(c). The lifting platform is a $z$ direction transition stage with maximal tuning range of 25~mm, which is larger than our cube (12.7~mm). We emphasize that the prisms in the second and third layers should be adjusted accordingly when generating different entanglement structures.  

The generated photon pairs have correlated polarization. In the ideal scenario, the polarization of these photon pairs is described by the maximally entangled two qubit state $\ket{\Psi^{+}_{ij}}=\frac{1}{\sqrt{2}}(\ket{H_oV_e}+\ket{V_eH_o})_{ij}$, where the subscript $o (e)$ represents the ordinary (extraordinary) component and $i, j$ denote the path label. Then, $\ket{\Psi^{+}_{ij}}$ is overlapped on PBS and becomes $\ket{\Phi^{+}_{ij}}=\frac{1}{\sqrt{2}}(\ket{H_oH_e}+\ket{V_oV_e})_{ij}$. In order to get a better indistinguishability, we filter the photons with proper full width at half of the transmittance maximum (FWHM) depending on whether it is an {\em o}-component  or an {\em e}-component light.

Specifically, in our experiment, the photons in path modes 1, 4, 6, 8 are {\em o}-component light and filtered by a narrow band filter with $\Delta_\text{\tiny FWHM}=4.6$~nm. The photons in path modes 2$^{\prime}$, 3$^{\prime}$, 5$^{\prime}$, 7$^{\prime}$, on the other hand, are {\em e}-component light and are filtered by a narrow band filter with $\Delta_\text{\tiny FWHM}=2.8$ nm. With such filter settings, we observe an eight-fold coincidences of 70/h in creating $\ket{G_{2222}}$. Each interference on PBS$_1$, PBS$_2$ or PBS$_3$ will reduce half of the eight-fold coincidence due to postselecting probability. We experimentally observed that the eight-fold coincidences in creating $\ket{G_{422}}$, $\ket{G_{44}}$, $\ket{G_{62}}$, and $\ket{G_{8}}$ are 36/h, 20/h, 20/h, and 8/h, respectively. The total eight-fold coincidences we collected in measuring $\langle\Mz\rangle)$, $\langle\Mx\rangle$, $\langle\A\rangle$, and $\langle\A^{\prime}\rangle$ are shown in Table~\ref{tab:rawdata}

\begin{table}[h!]
\centering
{
\begin{tabular}{|c|c|c|c|c|}
\hline
 State & $N(\Mz)$ & $N(\Mx)$ & $N(\A)$ & $N(\A^{\prime})$
\\ \hline
$\rho_8$ & 658 & 650 & 92 &  83 \\
$\rho_{62}$ & 260 & 240 & 260 & 176 \\
$\rho_{44}$ & 168 & 208 & 415 & 385  \\
$\rho_{422}$ & 196 & 232& 320 & 290 \\
$\rho_{2222}$ & 253 & 315 & 464 &  412 \\
\hline
\end{tabular}
\caption{Eight-fold coincidences in observing $\Mz$,  $\Mx$, $\A$, and $\A^{\prime}$ on $\rho_{8}$, $\rho_{62}$, $\rho_{44}$, $\rho_{422}$ and $\rho_{2222}$, respectively. The results of the calculated $\langle\Mz\rangle$,  $\langle\Mz\rangle$, $\langle\A\rangle$ and $\langle\A^{\prime}\rangle$ are shown in the main text. }
\label{tab:rawdata}}
\end{table}

\subsection{Imperfections and noise model}\label{noiseappendix}

\begin{figure*}[t!]
\centering
\includegraphics[width=2\columnwidth]{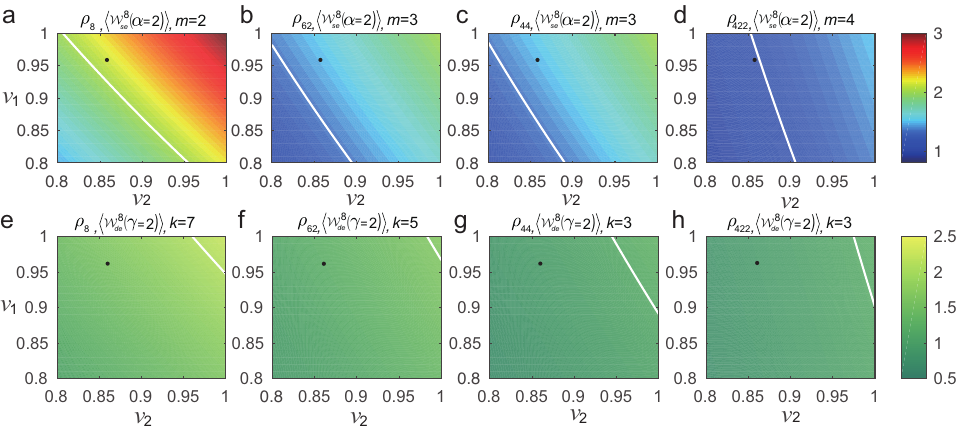}
\caption{\textbf{Calculation of the influence of the visibilities for our measurements of {$\ExpVal{\W_{se}^8(\alpha=2)}$ (top four plots) and $\ExpVal{\W_{de}^8(\gamma=2)}$} (bottom four plots) with respect to $\rho_{8}$, $\rho_{62}$, $\rho_{44}$ and $\rho_{422}$ (from left to right).} The straight white line represents combinations of $v_1$ and $v_2$, where the values of the witnesses for $m$-separability and $k$-producibility (for appropriate values of $m$ and $k$, respectively) are saturated [see Eq.~\eqref{Eq:vis:noise}]. The black circle marks the visibilities observed in our experiment, $v_1=0.967$ and $v_2=0.867$. }
\label{fig:visibility}
\end{figure*}

The experimental imperfections are (mainly) caused by the higher-order emissions in SPDC and the mode mismatch of the interference when superposing photons on PBS to connect entangled photon pairs. The influence of these imperfections can be reflected by the interference visibility. We define the visibility for an experimentally generated state as
\begin{equation}
v=\frac{\text{Target state}-\text{Noisy terms}}{\text{Target state}+\text{Noisy terms}}.
\end{equation}

For the entangled photon pair, the state can be written as
\begin{equation}
\rho_2=\frac{1+v_1}{2}\proj{\textrm{GHZ}_2}+\frac{1+v_1}{2}\frac{\id_2}{2^2},
\end{equation}
where $v_1$ represents the visibility of the entangled photon pair. The imperfections mainly come from the higher-order emission in SPDC. Similarly, the four-photon GHZ state can be written as
\begin{widetext}
\begin{equation}
\rho_4=\frac{1+v_2}{2}\left(\frac{1+v_1}{2}\right)^2\proj{\textrm{GHZ}_4}+\left[1-\frac{1+v_2}{2}\left(\frac{1+v_1}{2}\right)^2\right] \frac{\id_4}{2^4},
\end{equation}
\end{widetext}
where $v_2$ represents the visibility of interference on PBS. As two entangled photon pairs are involved when generating $\rho_4$,  the factor $(\frac{1+v_1}{2})^2$ is added. Generating an $n$-photon GHZ state requires ${n/2}$ entangled photon pairs and ${n/2-1}$ PBSs to connect them. Thus, an $n$-photon GHZ state, where $n$ is an even number, can be written as
\begin{widetext}
\begin{equation}\label{Eq:vis:noise}
\rho_n=\left(\frac{1+v_2}{2}\right)^{\frac{n}{2}-1}\left(\frac{1+v_1}{2}\right)^{\frac{n}{2}}\proj{\textrm{GHZ}_n}+\left[1-\left(\frac{1+v_2}{2}\right)^{\frac{n}{2}-1}\left(\frac{1+v_1}{2}\right)^{\frac{n}{2}}\right]\frac{\id_n}{2^n}.
\end{equation}
\end{widetext}
With this model, we calculate how the visibilities are related to our two witnesses. The calculations are shown in Fig.~\ref{fig:visibility}

The imperfections can also be modeled by noises. Our experimentally prepared $n$-photon state $\rho_{n}(n>2)$ can be represented as follows,
\begin{widetext}
\begin{equation}
{\rho_{n}=(1-\gamma_{d}^{(n)}-\gamma_{w}^{(n)})\proj{\textrm{GHZ}_n}+\frac{\gamma_d^{(n)}}{2}\left[\left(\proj{H}\right)^{\otimes n}+\left(\proj{V}\right)^{\otimes n}\right]+\gamma_w^{(n)}\frac{\id_{n}}{2^n}}.
\label{eq:GHZnoiseexp}
\end{equation}
\end{widetext}

In Eq.~\eqref{eq:GHZnoiseexp}, the first term describes the contribution from a genuine $n$-photon GHZ entangled state. The second term accounts for the imperfection of interference, which occurs with a probability of $\gamma_{d}^{(n)}$, where the polarized beam splitter does not superpose the photons from its two inputs. Experimentally, this is caused by the mode mismatch, including the mismatch of a narrow-band filter, the misalignment of the beams' direction  and other imperfections. The last term in Eq.~\eqref{eq:GHZnoiseexp} represents the higher-order emissions in SPDC processing, which is modeled by the white noise with corresponding probability  $\gamma_{w}^{(n)}$. The model we propose here is consistent with the observation that the expected value of $\mathcal M_Z^{8}$ is considerably than that of $\mathcal M_X^{8}$, which is common in the witness of a genuine multiphoton GHZ state based on the SPDC and photonic interferometer.

According to the model described in Eq.~\eqref{eq:GHZnoiseexp}, we can determine that the amount of noise our measurement of the witnesses $\ExpVal{\W_{se}^8(\alpha=2)}$ and $\ExpVal{\W_{de}^8(\gamma=2)}$ may tolerate with respect to the states $\rho_{8}$, $\rho_{62}=\rho_{6}\otimes\rho_{2}$, $\rho_{44}=\rho_{4}\otimes\rho_{4}$, $\rho_{422}=\rho_{4}\otimes\rho_{2}\otimes\rho_{2}$ and $\rho_{2222}=\rho_{2}\otimes\rho_{2}\otimes\rho_{2}\otimes\rho_{2}$. Note that in the calculation of $\rho_{62}$ and $\rho_{422}$, the contributions {$\gamma_{d}^{(2)}$} and {$\gamma_{w}^{(2)}$} from $\rho_2$ are negligible (and hence ignored) compared to the main contributions {$\gamma_{d(w)}^{(6)}$} and {$\gamma_{d(w)}^{(4)}$}  from $\rho_6$ and $\rho_4$. The calculated results are shown in Fig.~\ref{fig:noise1}.

\begin{figure*}[t!]
\centering
\includegraphics[width=2\columnwidth]{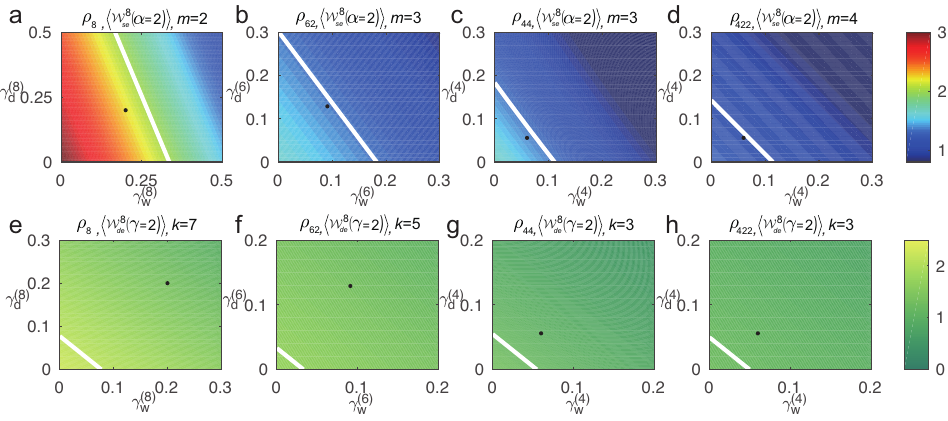}
\caption{\textbf{Calculation of the noise tolerance for our measurement of $\ExpVal{\W_{se}^8(\alpha=2)}$ (top four plots) and $\ExpVal{\W_{de}^8(\gamma=2)}$ (bottom four plots) with respect to $\rho_{8}$, $\rho_{62}$, $\rho_{44}$ and $\rho_{422}$ (from left to right).} The straight line represents combinations of noise parameters where the values of the witnesses for $m$-separability and $k$-producibility (for appropriate values of $m$ and $k$, respectively) are saturated, see Eq.~\eqref{eq:GHZnoiseexp}. Here, $\gamma_w^{(8)}$($\gamma_d^{(8)}$), $\gamma_w^{(6)}$($\gamma_d^{(6)}$), $\gamma_w^{(4)}$($\gamma_d^{(4)}$) denote the probability of white noise (decoherence noise) in preparing $\ket{{\rm GHZ}_8}$, $\ket{{\rm GHZ}_6}$ and $\ket{{\rm GHZ}_4}$. The $\gamma_w^{(2)}$ are negligible compare to $\gamma_w^{(4)}$ and $\gamma_w^{(6)}$, so  we do not consider them in the calculation of $\rho_{422}$, $\rho_{44}$ and $\rho_{62}$. For example, in the rightmost plots (for $\rho_{422}$), the straight line corresponds to the combination of $\gamma^{(4)}_d$ and $\gamma^{(4)}_w$ such that the 4-separable bound (upper plot) and the 3-producible (lower plot) bound are saturated. The black circle marks the noise parameters estimated from our measured value of $\ExpVal{\Mx}$ and $\ExpVal{\Mz}$. }
\label{fig:noise1}
\end{figure*}

\begin{figure}[h!]
\includegraphics[width=0.9\columnwidth]{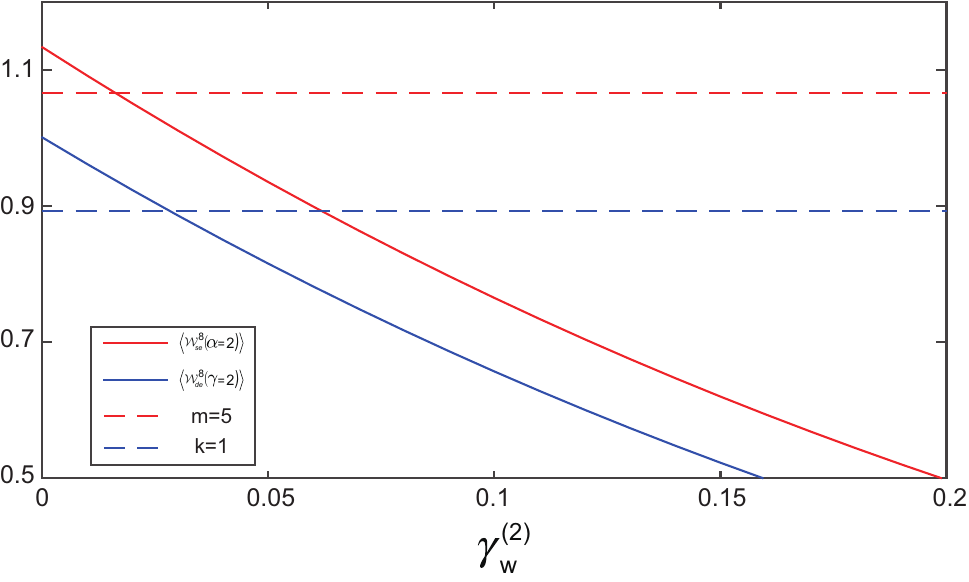}
\caption{\textbf{Calculation of the noise tolerance for our measurement of $\ExpVal{\W_{se}^8(\alpha=2)}$ and $\ExpVal{\W_{de}^8(\gamma=2)}$} on $\rho_{2222}$}
\label{fig:noise2}
\end{figure}

From Fig.~\ref{fig:noise1}, we observe that $\mathcal W_{se}$ can tolerate much more noise than $\mathcal W_{de}$, so experimentally $\mathcal W_{se}$ witnesses more precisely than $\mathcal W_{de}$. We also estimate {$\gamma_{d}^{(n)}$} and {$\gamma_{w}^{(n)}$} of $\rho_{n}$ by the measurements $\Mz$ and $\Mx$, and mark them in Fig.~\ref{fig:noise1} by a black circle. Note that $\gamma_{d}^{(n)}$ and $\gamma_{w}^{(n)}$ are related to $\langle \Mz\rangle$ and $\langle \Mx\rangle$ by
{\begin{equation}
\begin{split}
&\langle \Mz\rangle=\tr(\Mz\rho_{n})=1-\frac{2^{n-1}-1}{2^{n-1}}\gamma_{w}^{(n)},\\
&\langle \Mx\rangle=\tr(\Mx\rho_{n})=1-\gamma_{w}^{(n)}-\gamma_{d}^{(n)}.\\
\end{split}
\label{eq:estimation}
\end{equation}}
By measuring  $\langle \Mz\rangle$, $\langle \Mx \rangle$ and using Eq.~\eqref{eq:estimation}, we can calculate the values of $\gamma_{d}^{(n)}$ and $\gamma_{w}^{(n)}$.

For the state $\rho_{2222}$, the white noise model fits $\rho_{2}$ very well, and there is no interference between independent SPDC processes; thus, we simplify the noise model of Eq.~\eqref{eq:GHZnoiseexp} to consider solely the effect of white noise as $\rho_{2}=(1-\gamma_{w}^{(2)})\proj{\textrm{GHZ}_2}+\gamma_{w}^{(2)}\id_{2}/4$. The calculation results are shown in Fig.~\ref{fig:noise2}. Experimentally, a more accurate estimate of $\gamma_{w}^{(2)}=0.02$ is obtained by performing  tomographic measurements on $\rho_{2}$.\footnote{Strictly speaking, this value of $\gamma_{w}^{(2)}$ corresponds to that obtained by a better fitting of the experimental data to some $\ket{{\rm GHZ}_2(\theta,\phi)}$ (for $\theta\neq\tfrac{\pi}{4}$) state admixed with white noise.} Unlike the case of $\rho_{8}$, $\rho_{62}$,  $\rho_{44}$ and  $\rho_{422}$, $\mathcal W^8_{de}$ could tolerate a little bit more noise than $\mathcal W^8_{se}$ on $\rho_{2222}$. All the calculations hold under the assumption that the collection efficiencies in every mode are the same. The calculations need to be modified when the collection efficiencies are different.

\subsection{Algorithmic procedure to deduce {a minimal} entanglement structure }\label{procedureappendix}

In this section, we give a procedure to systematically deduce {a minimal} entanglement structure by using the results of $\langle \mathcal W_{se}\rangle$ and $\langle \mathcal W_{de}\rangle$. As shown in Fig.~\ref{fig:procedure}, we need to follow three steps to systematically deduce the underlying entanglement structure.

\noindent\textbf{Step 1:} For a given $n$-partite state $\rho$, we check whether it is genuinely $n$-partite entangled or not by measuring the witness $\mathcal W_{se}(\alpha)$. If it is, the task is completed, otherwise, we proceed to step 2.

\noindent\textbf{Step 2:} The extent to which the state is (not) $m$-separable for $m>1$ can be analyzed by using the measured value obtained in step 1 and considering the different $m$-separable bounds given in Eq.~\eqref{App:eq:separability}. Concurrently, we perform the measurement needed to evaluate $\ExpVal{\mathcal W_{de}(\gamma)}$. As with the case of separability, a lower bound on the entanglement depth can be obtained by analyzing the measured value against the various $k$-producible bounds.

\noindent\textbf{Step 3:} Based on the results from step 2, we can conclude that the entanglement intactness and entanglement depth are, respectively, upper and lower bounded by $m\leq M$ and $k\geq K$. From here, based on the data obtained during the measurement of $\ExpVal{\mathcal W^8_{se}(\alpha)}$ (and $\ExpVal{\mathcal W^8_{de}(\gamma)}$),  we may evaluate $\ExpVal{\mathcal W^{n'}_{se}(\alpha)}$ for all combinations of $n'\ge K$  parties to determine which among the $n$ parties exhibit genuine $K$-photon (or more-partite) entanglement, and which exhibit less-partite entanglement.

\begin{figure}[h!]
\includegraphics[width=\columnwidth]{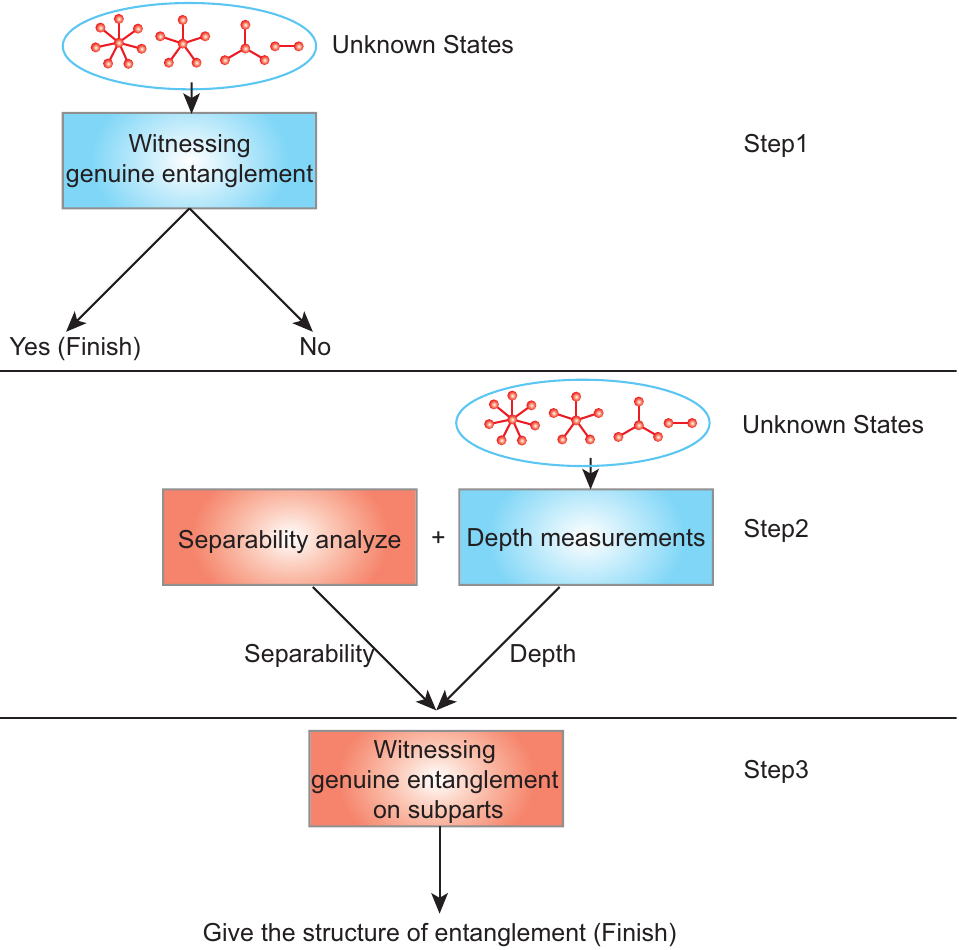}
\caption{\textbf{Procedure to deduce the underlying entanglement structure.}
Each orange box represents steps that only involve classical calculation (i.e., no consumption of any quantum resource is required). The blue box represents steps where both quantum resource and classical calculations are needed.}
\label{fig:procedure}
\end{figure}

\subsection{More experimental results for the deduction of {a minimal} entanglement structure }
\label{App:OtherStructures}

\begin{figure}[h!]
\includegraphics[width=\columnwidth]{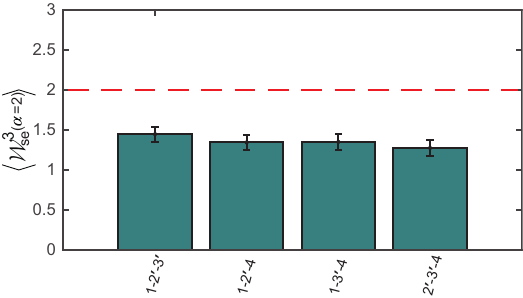}
\caption{\textbf{Estimated value of the three-body entanglement witness $\ExpVal{\mathcal W^3_{se}(\alpha)}$ for \{1, 2$^{\prime}$, 3$^{\prime}$, 4\} of $\rho_{422}$ based on the data acquired during the measurements of $\ExpVal{\mathcal W^8_{se}(\alpha)}$.}}
\label{fig:4223}
\end{figure}

The entanglement structure can be deduced by employing the procedure described in Appendix~\ref{procedureappendix}. We show the results for $\rho_{422}$ and $\rho_{2222}$ in the main text. {For state $\rho_{422}$, we show that photons in path mode \{5, 6$^{\prime}$, 7$^{\prime}$, 8\} are four-photon entangled [Fig.~\ref{fig:data2}(a) in the main text]. We omit the results of searching three-partite GMEs in \{1, 2$^{\prime}$, 3$^{\prime}$, 4\}. Below, we show that we our measurement of $\langle \mathcal W_{se}^{3}(\alpha=2)\rangle$ does not reveal any three-photon entanglement for any possible three-photon combination in  \{1, 2$^{\prime}$, 3$^{\prime}$, 4\}. The results are shown in Fig.~\ref{fig:4223}. We then search for two-partite entanglement. The results are shown in Fig.~\ref{fig:data2}(b) in the main text.}

For state $\rho_{62}$, the measurement result $\langle {{W_{de}^8(\gamma)}}\rangle={1.29\pm0.08}$ indicates that there is at least 4-photon entanglement in $\rho_{62}$.  So, we first try to identify the parties that exhibit this four-photon entanglement in $\rho_{62}$. As shown in Fig.~\ref{fig:62structure}(a), there are 7 four-photon combinations that violate the bipserapable bound of $\ExpVal{W_{de}^4(\gamma=2)}$, therefore indicating the presence of four-photon entanglement among these parties. However, the measured value of $\ExpVal{W_{de}^4(\gamma=2)}$ for the complementary set of parties does not reveal any four-photon entanglement . As $\langle {W_{se}^8(\alpha=4/3)}\rangle={1.43\pm0.07}$ indicates that $m\leq2$ for $\rho_{62}$, these results suggest that $\rho_{62}$ does not have the entanglement structure of $\ket{G_{44}}$. Similarly, the results in Fig.~\ref{fig:62structure}(b) suggest that $\rho_{62}$ does not have the entanglement structure of $\ket{G_{53}}$ either. Rather, an entanglement structure of $\rho_{62}$ that is compatible with our measurement results is that of $\ket{G_{62}}$ [Fig.~\ref{fig:62structure}(c)].

With the same procedure, according to the results $m\leq2$, $k\geq3$ (shown in the main text), we find that $\rho_{44}$ may not have the entanglement structure of $\ket{G_{35}}$. Instead, our results shown in Fig.~\ref{fig:44structure} suggest that one possible entanglement structure of $\rho_{44}$ is that given by $\ket{G_{44}}$.

\begin{figure*}[h!]
\includegraphics[width=1.9\columnwidth]{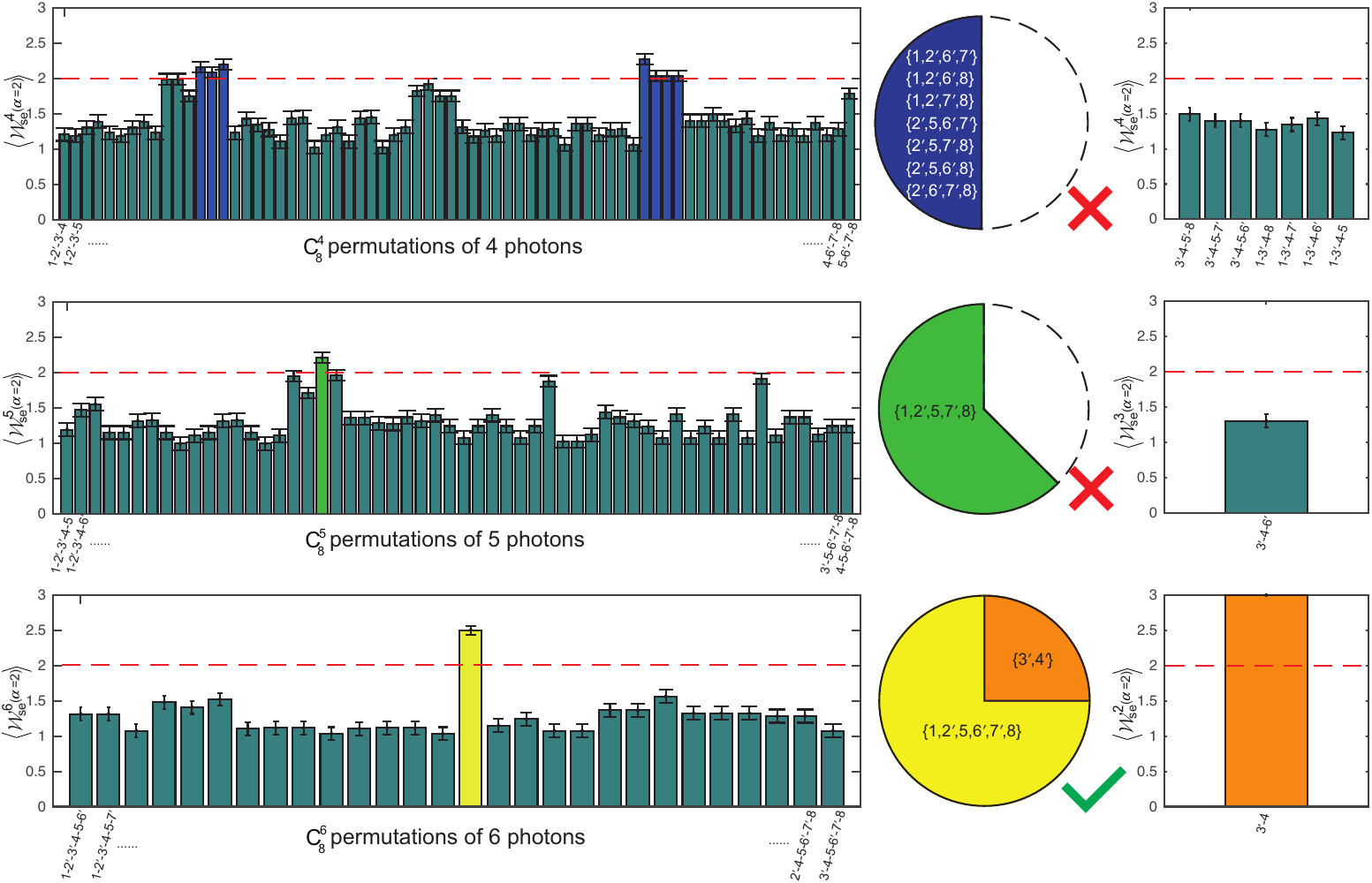}
\caption{\textbf{Results pertinent to the entanglement structure deduction for $\rho_{62}$}}
\label{fig:62structure}
\end{figure*}

\begin{figure*}[h!]
\includegraphics[width=1.9\columnwidth]{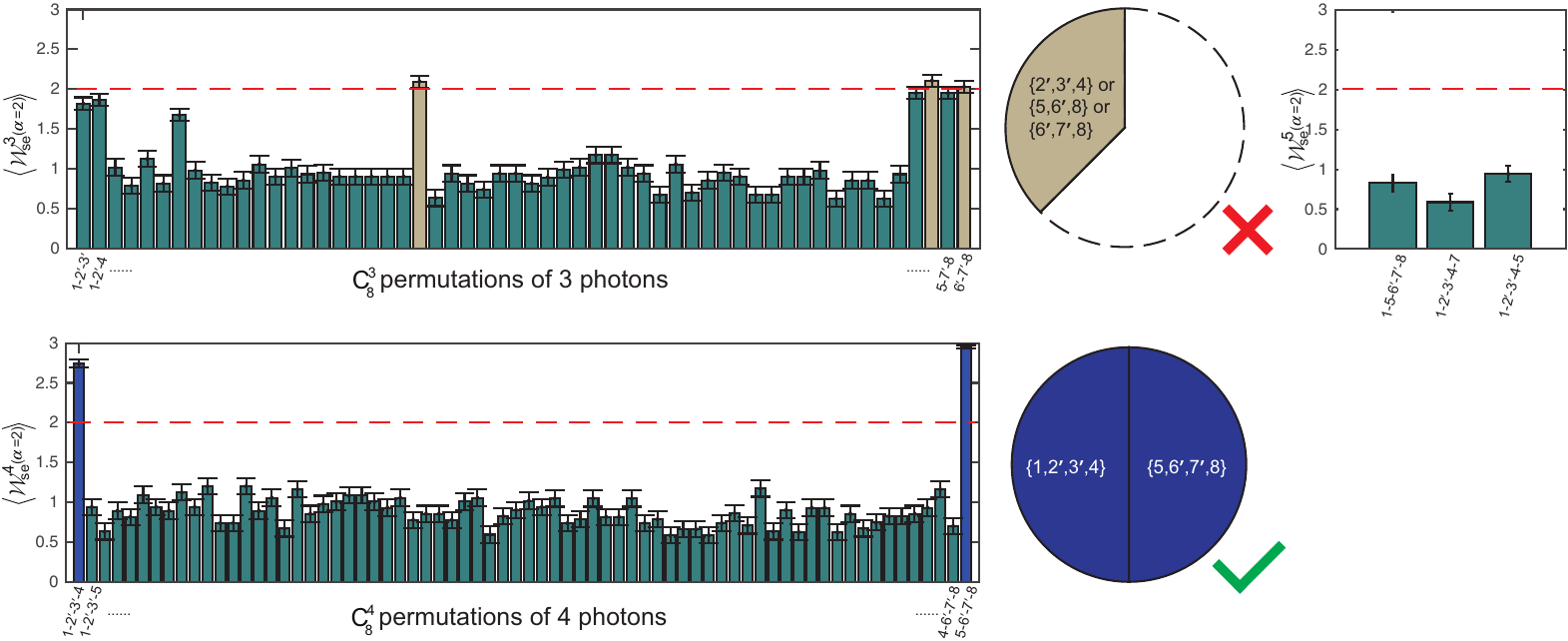}
\caption{\textbf{Results pertinent to the entanglement structure deduction for $\rho_{44}$}}
\label{fig:44structure}
\end{figure*}

\bibliography{MES}

\end{document}